\documentclass[journal]{IEEEtran}
\usepackage{amsmath,amsfonts,amssymb}
\usepackage{adjustbox}

\usepackage{amsthm,mdframed,calc}
\usepackage{mathrsfs}
\usepackage{cite}
\usepackage{textcomp,stfloats,url}
\usepackage{verbatim}
\usepackage{graphicx}
\usepackage{pifont}
\usepackage{tablefootnote}
\usepackage[usestackEOL]{stackengine}

\usepackage[boxed,linesnumbered,ruled]{algorithm2e}
\usepackage[skip=5pt,font={small},singlelinecheck=off]{caption}
\usepackage{subfig,float}
\usepackage{caption}
\usepackage{cases}

\usepackage{multirow}
\usepackage[draft]{hyperref}
\usepackage[normalem]{ulem}
\usepackage{empheq}
 \usepackage{tikz,lipsum,lmodern}
\usepackage{diagbox}
\usepackage[most]{tcolorbox}
\SetKwInput{KwInput}{Input}               
\SetKwInput{KwOutput}{Output} 
\SetKwComment{Comment}{/* }{ */} 
\SetAlgoCaptionSeparator{.}          
\ifCLASSINFOpdf

\else

\fi

\newtheorem{theorem}{Theorem}
\newtheorem{Proposition}{Proposition}

\newtheorem{problem}{Problem}

\begin{document}
\title{Helper-Friendly Latency-Bounded Mitigation Strategies against Reactive Jamming Adversaries}

\author{Soumita Hazra and J. Harshan
\thanks{The authors are with the Department of Electrical Engineering, Indian Institute of Technology Delhi, India. Email: Soumita.Hazra@ee.iitd.ac.in, jharshan@ee.iitd.ac.in.}}



\maketitle
\begin{abstract}
Due to the recent developments in the field of full-duplex radios and cognitive radios, a new class of reactive jamming attacks has gained attention wherein an adversary transmits jamming energy over the victim's frequency band and also monitors various energy statistics in the network so as to detect countermeasures, thereby trapping the victim. Although cooperative mitigation strategies against such security threats exist, they are known to incur spectral-efficiency loss on the \textcolor{black}{helper node}, and are also not robust to variable latency-constraints on victim’s messages. Identifying these research gaps in existing countermeasures against reactive jamming attacks, we propose a family of helper-friendly cooperative mitigation strategies \textcolor{black}{that are applicable for a wide-range of latency-requirements on the victim’s messages as well as practical radio hardware at the helper nodes}. The proposed strategies are designed to facilitate reliable communication for the victim, without compromising the helper’s spectral efficiency and also minimally disturbing the various energy statistics in the network. For theoretical guarantees on their efficacy, interesting optimization problems are formulated on the choice of the underlying parameters, followed by extensive mathematical analyses on their error-performance and \textcolor{black}{covertness}. Experimental results indicate that the proposed strategies should be preferred over the state-of-the-art methods when the helper node is unwilling to compromise on its error performance for assisting the victim.

\end{abstract}


\section{Introduction}
\label{sec:intro}

\IEEEPARstart{J}{amming} \cite{Jamming2022} is a popular form of Denial of Service (DoS) \cite{DOS2} attack on wireless communication wherein an adversary transmits jamming energy over the victim's frequency band such that the signal-to-interference-noise ratio (SINR) at the intended destination is insufficient for reliable communication. \textcolor{black}{Although frequency hopping based anti-jamming solutions have been well studied to mitigate such attacks, they are spectrally inefficient, making them inapplicable to bandwidth-hungry next-generation networks. To circumvent these problems, recent developments in the field of Reconfigurable Intelligent Surfaces (RIS) \cite{RIS_basics} have led to new spectrally-efficient anti-jamming techniques \cite{RIS_Anti_jam1, RIS_Anti_jam2, RIS_Anti_jam3, RIS_Anti_jam4} wherein reflective surfaces are used to control the direction and the strength of the transmitted signals in order to guarantee the required SINR level in the presence of jamming signals. These anti-jamming solutions are effective against non-reactive jammers that passively inject jamming energy aligned with the victim's beam without monitoring the victim's response.}

With the advent of full-duplex radios (FDRs) \cite{FDR3,SIC3} and the emergence of cognitive radios \cite{CR2015}, a new class of jamming adversaries called reactive adversaries, has received attention in the recent past, wherein the adversary has the capability to transmit jamming energy over the victim's frequency band while also monitoring the same band for detecting potential countermeasures.

Within the class of reactive adversaries, \cite{V1,V2,lowlatency} were the first to propose a threat model, wherein the reactive adversary transmits jamming energy over a victim's frequency while monitoring its \emph{average energy} to detect the presence of any countermeasures. As a consequence, it was shown that traditional mitigation strategies such as frequency-hopping could not be used under such threat models as it would drop the average energy level over victim's frequency band, thereby making the adversary suspicious. Recently, \cite{ISIT,TVT1} advanced the monitoring capability of the threat model in \cite{V1,V2,lowlatency} by considering a reactive adversary, which replaced the average energy detector by a Kullback-Leibler Divergence (KLD) estimator \cite{KLD} to detect any change in the statistical distribution on the energy of the received symbols on the victim's frequency band. Additionally, the threat model of \cite{ISIT,TVT1} considered multi-band sensing capability at the adversary, which could inject jamming energy on the victim's frequency band, and also monitor the following \emph{generalized} energy metrics on all the frequency bands: (i) the instantaneous energies of the observed symbols and (ii) the statistical distribution on the energies of the observed symbols. Thus, the threat model of \cite{ISIT,TVT1} has challenged wireless clients to reliably communicate to their destinations by maintaining strict covertness, i.e., without getting detected by a reactive adversary with \emph{generalized energy detectors}.

To reliably and covertly communicate under the presence of a reactive adversary with the above-mentioned generalized energy detectors, \cite{ISIT,TVT1} proposed a cooperative mitigation strategy called the Rate-Half Scheme (RHS), wherein the victim node takes the assistance from a nearby helper node to communicate its information to the destination. In the RHS, the victim node is asked to communicate on the helper's frequency band, wherein they adopt a form of time division multiple access (TDMA) to transmit their information symbols using a part of their energies. Concurrently, both the victim and the helper node cooperatively communicate dummy symbols on the victim's frequency band using their residual energies. The exact protocol for communication between the victim and the helper node is designed to ensure: (i) reliable communication of their information symbols, (ii) minimal disturbance in the statistical distribution of the transmitted symbols over the victim's and the helper's frequency band, and finally (iii) minimal disturbance in the instantaneous energy on the helper's frequency band. In the next subsection, we identify some research gaps in the RHS \cite{ISIT,TVT1}, and subsequently use these gaps to motivate our problem statement.

\color{black}
\begin{table*}[ht]\caption{\textcolor{black}{In this novelty table, Baseline 1, Baseline 2 and Baseline 3 are used to denote the work done in \cite{V2}, \cite{lowlatency} and \cite{TVT1}, respectively. Also, $M$ denotes the modulation order of the helper node and $q$ represents the number of channel uses.}}\label{Novelty_table}
{%
\centering
\begin{small}
 \begin{adjustbox}{max width=\textwidth}
\begin{tabular}{|c|c|c|c|c|}
\hline
\textcolor{black}{\textbf{Parameters}}     & \textcolor{black}{\textbf{Baseline 1}}       & \textcolor{black}{\textbf{Baseline 2}} & \textcolor{black}{\textbf{Baseline 3}} & \textcolor{black}{\textbf{Our contribution}} \\ \hline
\textcolor{black}{{Jammer's detection capability }} & \multicolumn{2}{c|}{\textcolor{black}{Average energy based detector}}    &  \multicolumn{2}{c|}{\begin{tabular}[c]{@{}l@{}}\textcolor{black}{- Instantaneous energy based detector}\\ \textcolor{black}{- KLD-estimator based detector}\end{tabular}} \\ \hline
   \textcolor{black}{Joint spectral-efficiency (bits/sec/Hz)}  & \textcolor{black}{$0.5 (1+\log_2M)$}   &  \textcolor{black}{$0.5 (1+\log_2M)$}   & \textcolor{black}{$0.25 (1+\log_2M)$}  &  \textcolor{black}{$0.5 (0.5+\log_2M)$}     \\ \hline           
 \textcolor{black}{FDRs}  & \textcolor{black}{Impractical} &  \textcolor{black}{Practical  }& 
 \multicolumn{1}{c|}{\textcolor{black}{Impractical}} & 
 \begin{tabular}[c]{@{}l@{}}\textcolor{black}{Considers practical FDRs}\\ \textcolor{black}{(Delay-Tolerant Rate-Three-Fourth Scheme)}
 \end{tabular}   \\ \hline           
\textcolor{black}{ Strict latency constraints on victim's messages}  & \textcolor{black}{Applicable} & \textcolor{black}{Inapplicable}  & \multicolumn{1}{c|}{\textcolor{black}{Inapplicable}} & \begin{tabular}[c]{@{}l@{}}\textcolor{black}{Applicable for strict latency constraints on victim's messages}\\ \textcolor{black}{(Low-Latency Constrained Rate-Three-Fourth Scheme)} \end{tabular}   \\ \hline
      \textcolor{black}{Secret-key overhead} & \textcolor{black}{{0}}  & \textcolor{black}{{$2$  bits/block}}&\multicolumn{1}{c|}{\textcolor{black}{$1+ 2\log_2 M$  bits/block}} & \textcolor{black}{$1+ \log_2 M$  bits/block}        \\ \hline           
      \textcolor{black}{Modulation at helper} &  \multicolumn{3}{c|}{\textcolor{black}{Only $M$-PSK}}   & \textcolor{black}{Arbitrary}      \\ \hline                 \textcolor{black}{Position of helper} &  \multicolumn{3}{c|}{\textcolor{black}{Specific}}   & \textcolor{black}{Arbitrary}      \\ \hline     
       \textcolor{black}{Reliability} &\multirow{3}{*}{\begin{tabular}[c]{@{}l@{}}\textcolor{black}{Inapplicable due} \\ \textcolor{black}{to impractical FDR} \end{tabular}}  & \textcolor{black}{Inferior} &\multicolumn{1}{c|}{\textcolor{black}{Superior (Victim-friendly)}} & \textcolor{black}{Superior (Helper-friendly)}        \\ \cline{1-1} \cline{3-5} 
        \textcolor{black}{Covertness}  &    & \textcolor{black}{Superior} &\multicolumn{1}{c|}{\textcolor{black}{Inferior}} & \textcolor{black}{Inferior}       \\ \cline{1-1} \cline{3-5} 
          \textcolor{black}{Joint reliability and covertness at operating point}  &    & \textcolor{black}{Inferior} &\multicolumn{1}{c|}{\textcolor{black}{Inapplicable}} & \textcolor{black}{Superior }     \\ \hline  
         \textcolor{black}{Spectral-efficiency of helper node} & \multicolumn{2}{c|}{\textcolor{black}{$\log_2M$}}     & \multicolumn{1}{c|}{\textcolor{black}{$0.5\log_2M$}} & \textcolor{black}{$\log_2M$}       \\   \hline   
\textcolor{black}{Decoding complexity (Optimal)}  & \textcolor{black}{$\mathcal{O} (2M)$}   & \textcolor{black}{$\mathcal{O} (4M)$} &\multicolumn{1}{c|}{\textcolor{black}{$\mathcal{O} (2M)$}} & \textcolor{black}{DTRTF-   $\mathcal{O} (2M^2)$,  LLCRTF-   $\mathcal{O} (2M)$}    \\ \hline
\textcolor{black}{Decoding complexity  (Sub-Optimal)}& \textcolor{black}{$\mathcal{O} (2M)$}   & --- &\multicolumn{1}{c|}{\textcolor{black}{$\mathcal{O} (2M)$} }& \textcolor{black}{DTRTF-   $\mathcal{O} (2M)$, LLCRTF-   ---}        \\ \hline
  \begin{tabular}[c]{@{}l@{}}
  \textcolor{black}{Space complexity for decoder}\end{tabular}& \textcolor{black}{$0$} & \textcolor{black}{$0$} &\multicolumn{1}{c|}{\textcolor{black}{$1$}} & \textcolor{black}{DTRTF-   $q$,  LLCRTF-   $0$}    \\ \hline
         
\end{tabular}%
  \end{adjustbox}
\end{small}

}
\end{table*}
\color{black}

\subsection{Motivation and Problem Statement}

As part of the RHS, although the victim and the helper node adopt a form of TDMA for transmitting their information symbols, the helper node is asked to transmit a dummy symbol preshared with the destination in order to maintain instantaneous energy over the helper's frequency band. As a consequence, this strategy reduces the transmission rate of the helper node by half, and moreover, requires the helper node and the destination to agree upon a shared randomness of $\log_2 M$ bits per block, where $M$ is the modulation order of the helper node. Besides the above mentioned drawbacks, we also identify that the RHS does not cater to victims that have variable latency constraints on their messages, i.e., when the victims have the constraint that their messages must be decoded at the base station within a certain deadline.  Thus, considering the practical aspects of reduced transmission rate, high secret-key overhead and rigid latency constraints, we pose the following research questions in this work: \emph{How to design cooperative countermeasures against a reactive adversary that is equipped with generalized energy detectors
1) without compromising the rate of the helper node,
2) with low secret-key overhead for \textcolor{black}{covert communication},
3) \textcolor{black}{with practical hardware at the helper node}, and
4) supporting variable latency-constraints at the victim node?
}

\subsection{Contributions}

For the threat model involving a reactive adversary \cite{ISIT,TVT1}, we present two classes of low secret-key overhead, helper-friendly cooperative mitigation strategies \textcolor{black}{that can jointly accommodate a wide range of latency-constraints on the victim's messages as well as practical radio hardware at the helper nodes}. Instead of asking the victim and the \textcolor{black}{helper node} to opt for a form of TDMA structure as in \cite{ISIT,TVT1}, our methods ask the two nodes to adopt a power-controlled uplink multiple access structure, thereby not incurring rate-loss on the helper node (see Section \ref{SM}). Our specific contributions are listed below. 
 
1) For the regime \textcolor{black}{when the latency constraints on the victim's messages are relaxed compared to the \emph{reaction-time} of the helper node, we propose a full-duplex radio based strategy referred to as the Delay-Tolerant Rate-Three-Fourth (DTRTF) scheme}, wherein the victim and the \textcolor{black}{helper node} cooperatively pour their energies on each other's frequency band to minimally disturb the energy statistics in the network (see Section \ref{DTMS}). Subsequently, we formulate an optimization problem on the energy parameters of the DTRTF scheme so as to achieve both reliable and covert communication with the destination. While the \textcolor{black}{covertness} analysis for DTRTF scheme follows from the contributions in \cite{ISIT,TVT1}, the reliability analysis requires a fresh look owing to the underlying multiple-access frame structure. Therefore, we present a comprehensive error analysis of the DTRTF scheme (see Section \ref{DecodingatBob}), \textcolor{black}{when the helper node uses Phase Shift Keying (PSK) to communicate its information to the destination.}{\footnote{\textcolor{black}{For the adversary with the instantaneous energy detector, PSK provides the worst-case scenario in terms of covertness, owing to the fact that symbols modulated using PSK have a constant energy level. In the later section, we also discuss the performance of the proposed strategies for quadrature amplitude modulation.} }

2) To cater to the case \textcolor{black}{when the victim node has rigid latency constraints compared to the \emph{reaction-time} of the helper node, we propose the Low-Latency Constrained Rate-Three-Fourth (LLCRTF) scheme which does not require a full-duplex radio at the helper node}. Similar to the DTRTF scheme, interesting optimization problems are formulated to establish reliable and covert communication with the destination. While the framework of the LLCRTF scheme is similar to that of the DTRTF scheme, its error analysis is inherently different. As a result, we present a comprehensive error analysis associated with the decoding of the information symbols transmitted by the helper and the victim for the LLCRTF scheme (see Section \ref{LLCA}).

3) Using the reliability and the \textcolor{black}{covertness} analyses on the DTRTF and the LLCRTF schemes, we present low-complexity algorithms to solve the optimization problems on their energy parameters. Subsequently, using the proposed solutions, we show that our schemes can facilitate reliable communication for the victim while maintaining the required energy statistics such as the instantaneous energy and the statistical distribution on the energies of the symbols on all the frequency bands. (see Section \ref{CD}). Finally, when implementing the DTRTF and LLCRTF schemes, we use extensive simulations to compare the individual error performance of the helper and the victim nodes against the state-of-the-art RHS. Our results indicate that while RHS is a victim-friendly scheme, the proposed schemes are more helper-friendly in terms of their error performance (see Section \ref{SR}). \textcolor{black}{Lastly, we discuss the  robustness of the proposed strategies by considering QAM  at the helper node, and arbitrary position of the helper node with respect to the victim and the destination (see Section \ref{Robustness}).}

\textcolor{black}{Overall, the novelty of this work is highlighted in Table \ref{Novelty_table}, wherein the salient features of our contributions are compared with that of \cite{V2, lowlatency, ISIT, TVT1}. From this table, we observe that our proposed strategies outperform the state-of-the-art strategies in most of the listed parameters.} \textcolor{black}{Although a rich set of literature on RIS assisted anti-jamming techniques \cite{RIS_Anti_jam1, RIS_Anti_jam2, RIS_Anti_jam3, RIS_Anti_jam4} exist, these countermeasures are not relevant baselines for our work as they are not designed against reactive adversaries that monitor the victim's countermeasures in terms of long-term and short-term energy statistics. Moreover, those anti-jamming techniques are not applicable when the jammer is placed close to the base-station and the jamming energy is injected in an omnidirectional manner. In contrast, our countermeasures continue to be effective against an omnidirectional jammer as the information bearing signals are bypassed through the frequency band of a helper node.} Our contributions stem from some notable developments in the field of methodologies devised for Full Duplex Radios (FDRs) \cite{SIC1,SIC2,SIC2020}, cognitive radios with FDRs \cite{FDCR1,FDCR2018,FDCR2,FDCR2021} and jamming attacks on FDRs \cite{FDJ1,FDJ2}. A preliminary version of this work is available in \cite{NCC}, which presents a basic variant of the DTRTF scheme without any analysis.

\begin{figure*}%
\vspace{-0.8cm}
\begin{center}
    \subfloat[]
{\includegraphics[scale=0.40]{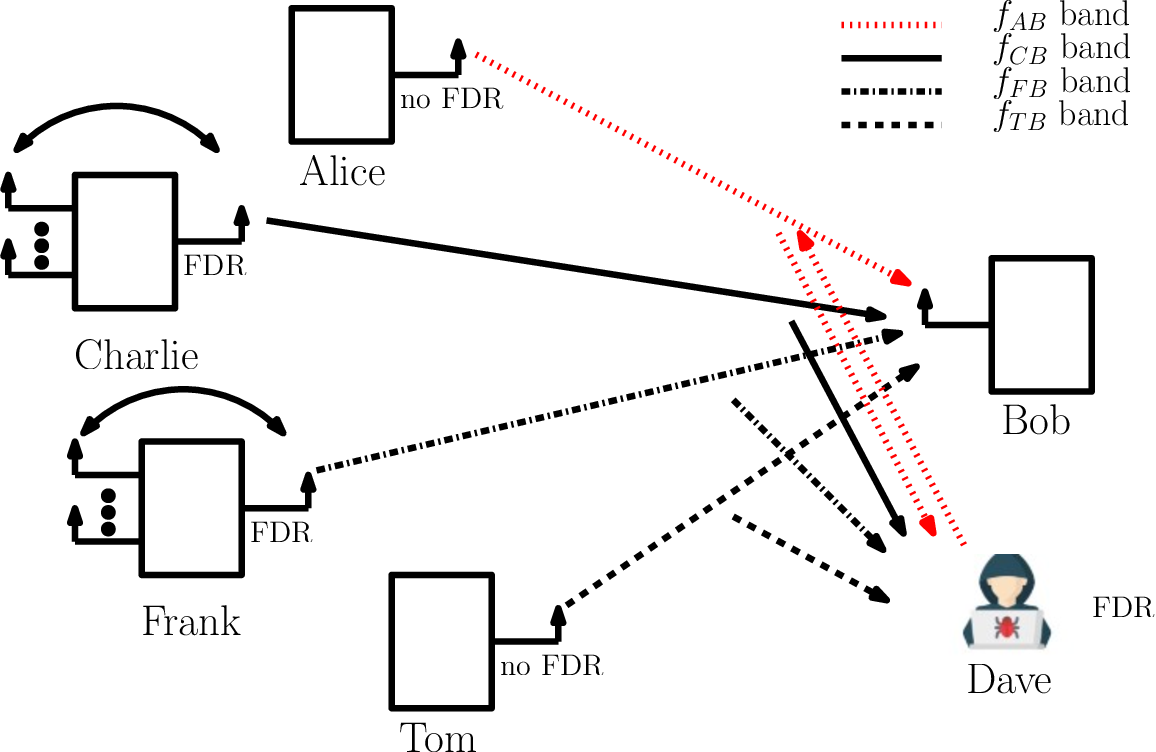}\label{SMdiag}
}
    \hfil
    \subfloat[]{
{\includegraphics[scale=0.40]{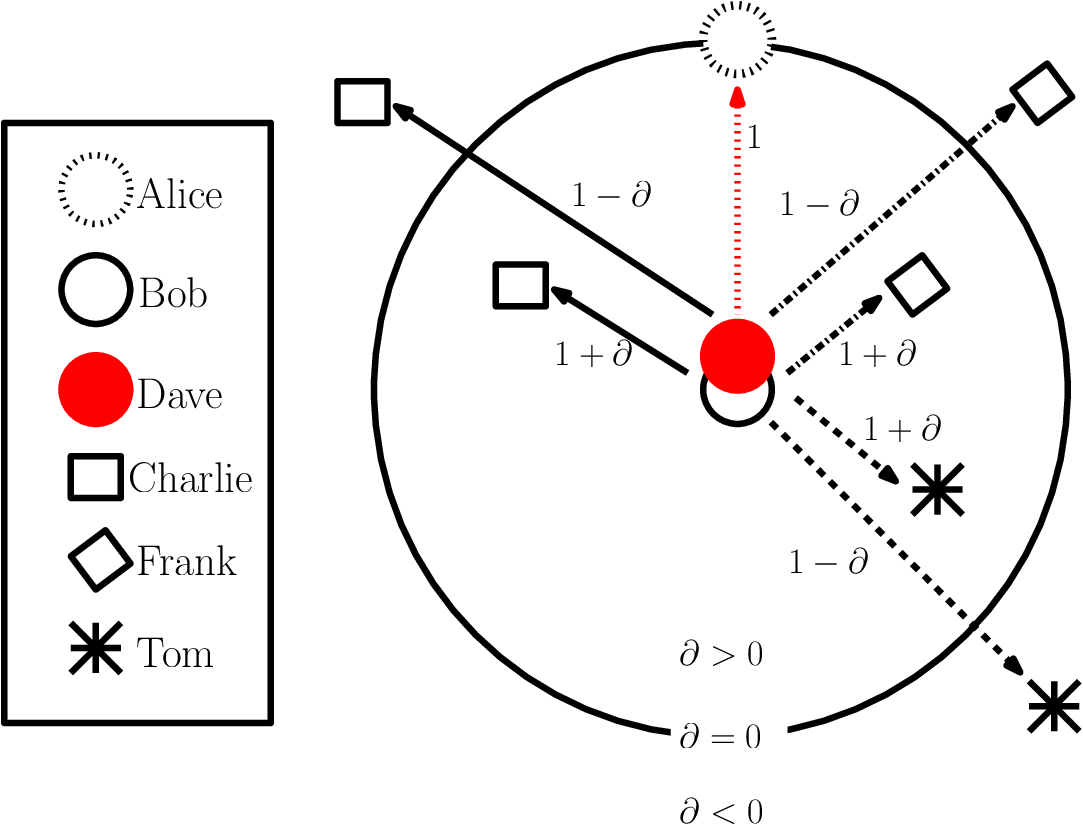}\label{Geometry}}
    \label{geometry}}%
    
    \caption{\textcolor{black}{(a)  Network model depicting uplink communication between the UEs and the base station. The reactive adversary, namely, Dave, is seen injecting jamming energy on the victim Alice, while monitoring all the network frequency bands.
    (b) Geometry capturing the relative positions of the UEs with respect to Bob, wherein, $\partial$ captures the variable for large-scale fading. \textcolor{black}{All the channels experience multi-path fading due to mobility which is captured through small-scale fading.}}}
    \label{Systemmodel}
\end{center}
\vspace{-0.3cm}
\end{figure*}

\color{black}
\section{Network Model} \label{SM}

We consider a single-cell based wireless network, as shown in Fig. \ref{Systemmodel}, wherein a number of user equipments (UEs) communicate with the base station through their uplink channels on orthogonal frequency bands. We assume that the network is crowded and heterogeneous, i.e., all the uplink frequencies of the base station are allocated to the UEs, and the UEs have arbitrary rate requirements for their uplink channels. One such instantiation of the network consists of four UEs, namely, Alice, Charlie, Tom and Frank, that communicate with the base station, namely, Bob \textcolor{black}{over a bandwidth of $W$ Hertz} on the \textcolor{black}{center} frequencies $f_{AB}$, $f_{CB}$, $f_{TB}$ and $f_{FB}$, respectively. The \textcolor{black}{average energies of the UEs per channel use are respectively denoted by $E_{A}, E_{C}, E_{T},$ and $E_{F}$}. 
\textcolor{black}{
We assume that the uplink channels of Alice, Charlie, Tom and Frank are frequency-flat and are quasi-static, with a coherence-time of $\tau_A$, $\tau_C$ ,
$\tau_T$ and $\tau_F$ time-slots, respectively.} Henceforth, one time-slot refers to the time taken for one channel use, which is $\frac{1}{W}$ seconds. Among the four UEs, Alice has critical information to communicate with Bob, which may be delay-tolerant or have strict low-latency constraints to reach the base station. Therefore, she is a potential victim of a DoS attack from an adversary. Also, given that coherent communication for Alice might result in pilot contamination attacks from adversaries, we assume that she communicates with Bob using non-coherent On-Off keying (OOK) modulation. 
\textcolor{black}{This implies that Bob decodes Alice's information symbols without using the instantaneous channel state information (CSI).} 
On the other hand, Charlie, Tom and Frank, which are not directly under adversarial attack, have strict high-rate requirements on their messages, and they communicate with Bob using coherent modulation schemes.
To support coherent demodulation on the $f_{CB}$, the $f_{TB}$ and the $f_{FB}$ bands, we assume that Charlie, Tom and Frank broadcast pilot symbols \cite[Section 6.4.1.1]{3gpp} followed by information symbols within the coherence-time, such that \textcolor{black}{$\tau_{C} = \tau_{C}^{p}+\tau_{C}^{d}$, $\tau_{T} = \tau_{T}^{p}+\tau_{T}^{d}$ and $\tau_{F} = \tau_{F}^{p}+\tau_{F}^{d}$, where $\tau_{C}^{p}$, $\tau_{T}^{p}$ and $\tau_{F}^{p}$, denote the number of time-slots for pilots, and $\tau_{C}^{d}$, $\tau_{T}^{d}$ and $\tau_{F}^{d}$, denote the number of time-slots for data, by Charlie, Tom and Frank, respectively. Here $\tau_{C}^{p}$, $\tau_{T}^{p}$ and $\tau_{F}^{p}$ are chosen such that the corresponding CSI are estimated with reasonable accuracy \cite{PCC1, PCC2}.}   \textcolor{black}{Also, Charlie and Frank have requirement of high spectral-efficiency on their uplink frequencies, as a result, they are equipped with practical FDRs with multiple receive antennas and single transmit antenna, however, there is no such requirement at Tom, therefore, he is not equipped with FDR.}

\subsection{Threat Model} \label{AM}

\textcolor{black}{In the above network, we also assume the presence of a reactive adversary, namely Dave, with single transmit antennas, that intends to execute a DoS threat on the critical messages of Alice. Similar to the threat model in \cite{ISIT,TVT1}, we assume that Dave, who is equipped with an FDR, has the knowledge of the uplink frequency of Alice. As a result, he transmits jamming energy over the $f_{AB}$ band and also monitors the various frequency bands in the network for any possible countermeasures from Alice.}\footnote{\textcolor{black}{Dave is equipped with a perfect FDR, implying perfect SIC. This is to assume a worst-case scenario for covertness.}} In particular, Dave uses the following tools to detect the presence of countermeasures: (i) On every frequency band in the network, he compares the statistical distribution on the energies of the received symbols before and after the attack using a KLD-estimator. (ii) He has complete knowledge of the modulation schemes used by the different UEs, and therefore, he monitors the instantaneous energy of the transmitted symbols over different frequencies of the network by using the CSI between the UEs and himself, \textcolor{black}{which is obtained by listening to the pilots within each coherence block.}
\textcolor{black}{Subsequently, on detecting variations in the generalized energy statistics on the $f_{AB}$ band, Dave suspects a possible countermeasure and then will pour all his jamming energy randomly on one of the other bands. Conversely, if he detects variations in the energy statistics on any other band, he will suspect a countermeasure, and then will pour all his jamming energy on that band. Consequently, in both the scenarios, the communication of other UEs with Bob will get affected,  
as a result, countermeasures should be designed in such a way that the probability of getting detected is minimum, and Dave is engaged only to the $f_{AB}$ band.}

\color{black}
\subsection{Overview of the Proposed Strategies}

To assist reliable and covert communication for Alice, we propose a family of cooperative mitigation strategies wherein Alice takes the assistance from a nearby helper node, \textcolor{black}{who is already using coherent-modulation on its band. First, we ask ask Alice to move to the frequency band of the helper node and transmit her information symbols on the time-slots when the helper is communicating its information.} Subsequently, \textcolor{black}{we ask the incumbent user of that frequency band to listen to her information symbols using an FDR and forward the same to Bob in a different time-slot provided the victim's latency constraints can be met.} The two users coordinate their energy resources such that the short-term and the long-term energy statistics of the underlying frequency bands are minimally disturbed. \textcolor{black}{We highlight that this cooperation between Alice and the helper takes place only in the data transmission phase, while the pilot transmission phase remains unaltered.} \textcolor{black}{In our network model, we can choose either Charlie or Frank as the helper node since both are equipped with FDRs and they are capable of listening to Alice's information symbols while they transmitting their information symbols. Henceforth, \textcolor{black}{in this paper, we denote the frequency band of the helper, the average energy of the helper, the number of time-slots in a coherence-block, number of pilots per coherence-block, number of information symbols per coherence-block of the helper channel by $f_{HB}, E_{H}, \tau_{H}, \tau^{p}_{H}$ and $\tau^{d}_{H}$, respectively, where $f_{HB} \in \{f_{CB}, f_{FB}\}$, $E_H \in \{E_C, E_F \}$, $\tau_H \in \{\tau_C, \tau_F \}$, $\tau^{p}_H \in \{\tau^{p}_C, \tau^{p}_F \}$ and $\tau^{d}_H \in \{\tau^{d}_C, \tau^{d}_F \}$}.}

In contrast to prior contributions, our model considers practical FDRs at the helper node, which incur a delay of several time-slots, to achieve reasonable accuracy in their self-interference cancellation (SIC) modules. \textcolor{black}{To formally capture this delay metric, earlier referred to as \emph{reaction-time} of FDRs}, we define \textcolor{black}{$n_{fd} \in \mathbb{Z}_{+}$} as the maximum number of time-slots required by the helper node to keep the residual interference $\rho$ of his SIC block, below a certain threshold, denoted by $\rho_{th}$, \textcolor{black}{needed for reliably decoding Alice's information symbols.} \textcolor{black}{Given that pilot symbols are transmitted within a coherence-block, Alice's decoded information symbol cannot always be incorporated in the helper's frame right after $n_{fd}$ time-slots. With $\tau_{H}$ time-slots for a coherence-block, it is straightforward to verify that Alice's information symbol can be incorporated with a delay of $n_{fr} \triangleq (\lceil \frac{n_{fd}}{\tau_{H}}\rceil + 1) \tau_{H}$ time slots. Thus, $n_{fr}$ characterizes the effective delay offered by the FDR at the helper node.} Furthermore, given that Alice may have low-latency requirements on her messages, we define the parameter $m \in \mathbb{Z}_{+}$, to capture the maximum number of time-slots for which Bob can wait to decode Alice's symbol. Thus, to cater to diverse requirements on the delay offered by the helper node, and the low-latency constraints on Alice, we design mitigation strategies that are applicable for different combinations of $m \geq 0$ and $\textcolor{black}{n_{fr}} > 0$.
For $m \geq \textcolor{black}{n_{fr}}$ and $m \neq 0$, i.e., when the delay introduced by by the helper node to decode Alice's bits is less than or equal to the maximum number of time-slots for which Bob can wait for Alice's bits, we propose the Delay-Tolerant Rate-Three-Fourth (DTRTF) scheme. On the other hand, for $m < \textcolor{black}{n_{fr}}$ and $m \geq 0$,\footnote{Note that $m=0$ represents the case when Bob wants to decode Alice's messages instantly upon observing the received symbols.} we propose the Low-Latency Constrained Rate-Three-Fourth (LLCRTF) scheme. \textcolor{black}{We highlight that both DTRTF and LLCRTF schemes are applicable for quasi-static channels with arbitrary coherence-time as long as $\tau_H^p$ is sufficient enough to achieve coherent communication, and in such a case, the minimum coherence-time is $\tau_H^p + 1$.} \textcolor{black}{Suppose $\tau_H^p=1$, which captures a quasi-static channel with coherence-time 2 time-slots, Bob uses one pilot symbol to obtain CSI, and in such a case, the obtained CSI would be inaccurate resulting in degraded error performance.}

\begin{figure}[ht!]
   \begin{center}
       {\includegraphics[scale=0.35]{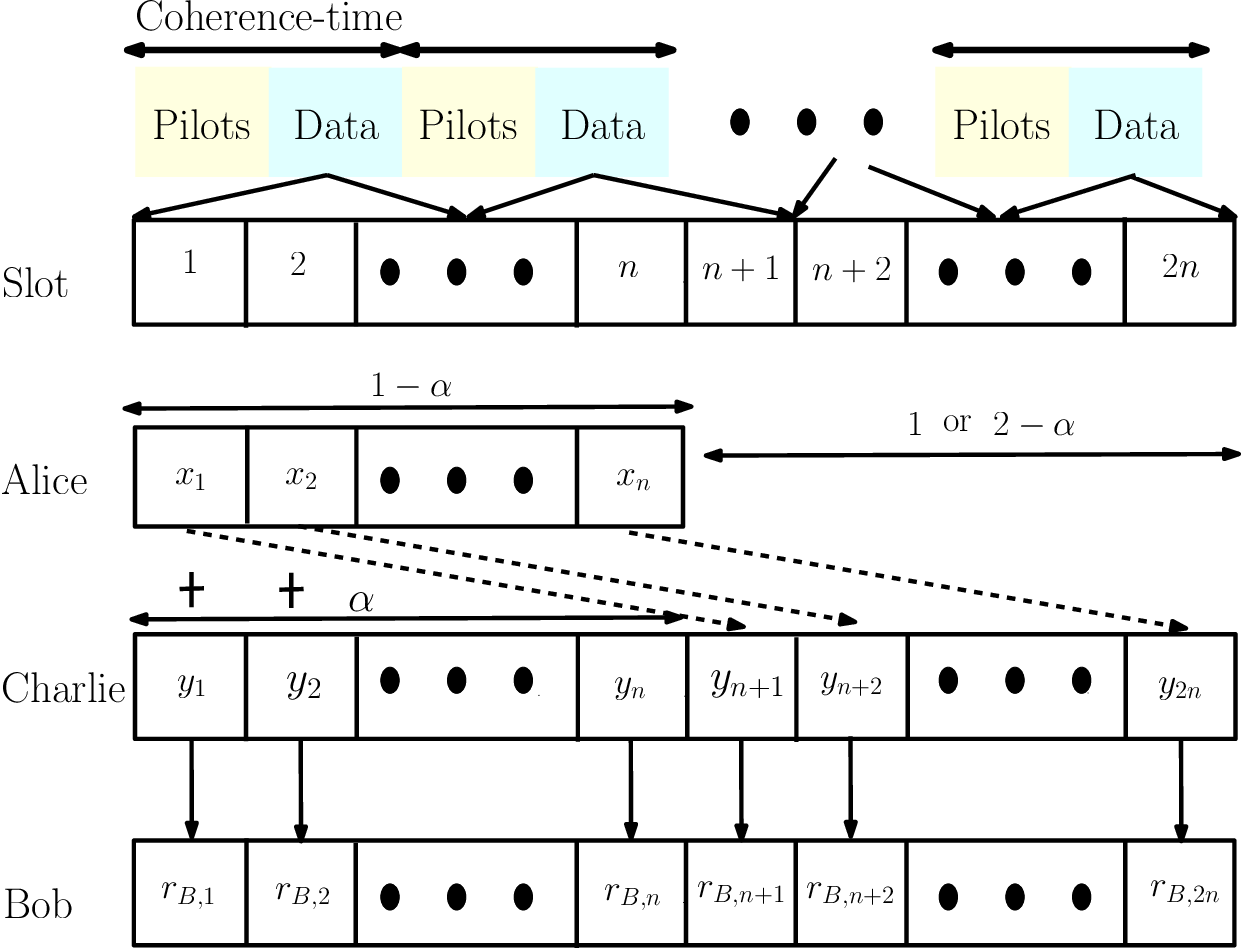}}
   \end{center}  
\caption{\textcolor{black}{Depiction of the frame structure and the corresponding energy levels over $f_{HB}$ in the DTRTF scheme. \textcolor{black}{Only the data transmission phase is captured for exposition.}}}
\label{nfcbdiag}
    \end{figure}

 \begin{figure*}
 \vspace{-0.8cm}
 \begin{center}
     \subfloat[]
{
\includegraphics[scale=0.42]{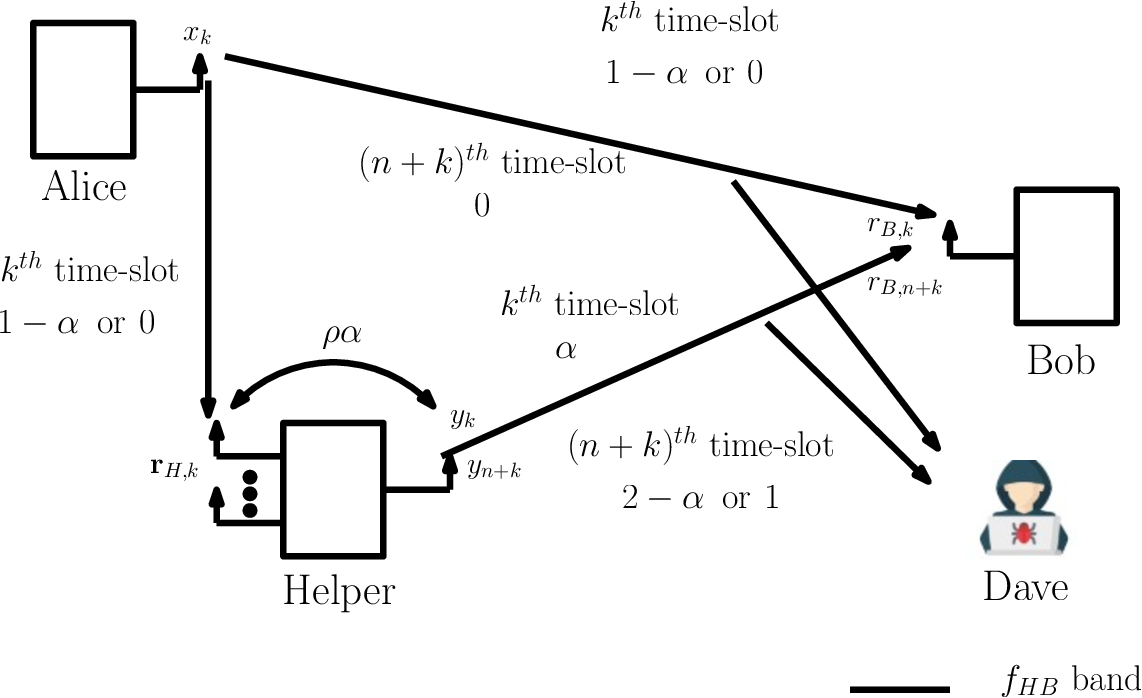}
    \label{fcbkn}}%
    \hfil
    \subfloat[]
{
\includegraphics[scale=0.42]{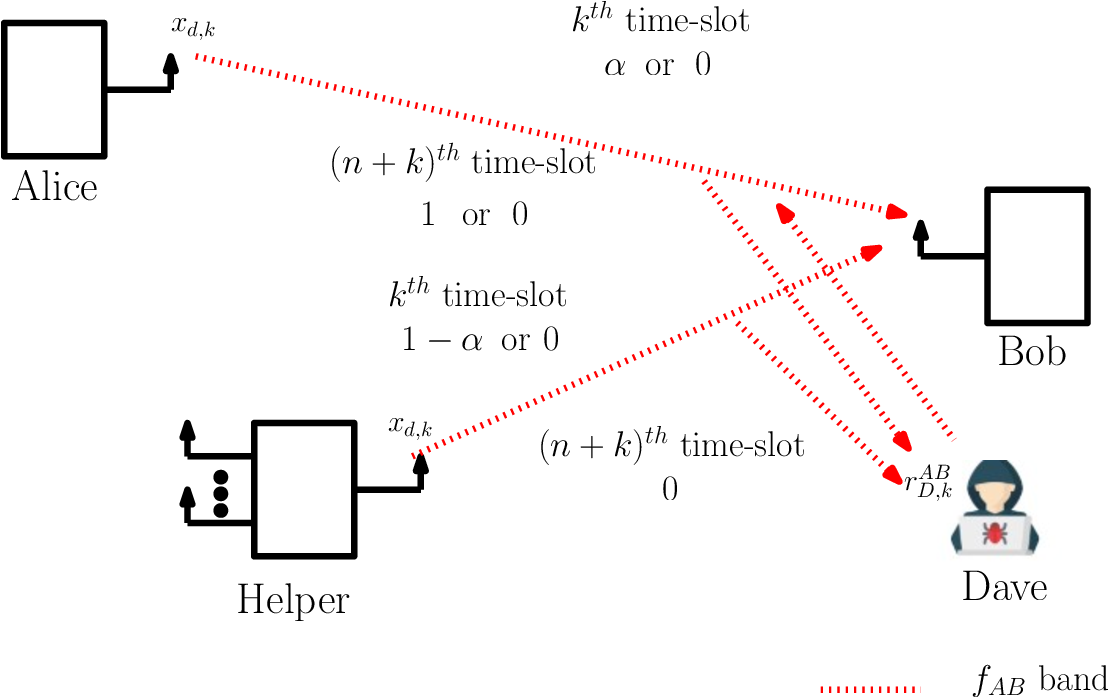}
    \label{fabntimeslots}}%
 \end{center}

    \vspace{-0.3cm}
    \caption{ 
    \textcolor{black}{Representation of the transmission during the $k^{th}$ and $(n+k)^{th}$ time-slot over (a) $f_{HB}$ band, and (b) $f_{AB}$ band.}
    }
    \label{fcb2diags}
\end{figure*}

\section{DTRTF Mitigation Strategy} \label{DTMS}
In this section, we propose the DTRTF scheme,  to cater to the case when $m\geq \textcolor{black}{n_{fr}}$. This strategy allows Bob to decode Alice's bits with a delay of $\textcolor{black}{n_{fr}}$ time-slots. 
\textcolor{black}{Owing to the fact that the pilot transmission phase remains unaltered and the data transmission phase undergoes changes in the proposed strategies, we will only present the data transmission phase of the DTRTF strategy. As shown in Fig. \ref{nfcbdiag}, we will exclude the pilots and only focus on the time-slots for data transmission to aggregate a set of $2n$ symbols on the helper's band. Here $n =\left \lceil \frac{n_{fd}}{\tau_{H}} \right \rceil \tau^{d}_{H}$ refers to the number of helper's information symbols for which Alice's decoded information symbols cannot be incorporated due to the effective delay $n_{fr}$ of the FDR.} Therefore, during the first $n$ time-slots, Alice and the helper form an uplink multiple access channel to transmit their information symbols using a portion of their energies over the $f_{HB}$ band. Meanwhile, they cooperatively pour their remaining energies over the $f_{AB}$ band to generate a pseudorandom OOK sequence using a preshared key. In the subsequent $n$ time-slots, the helper transmits a modified version of his  symbols on the $f_{HB}$ band by incorporating Alice's bits that were decoded $n$ time-slots ago using his FDR. Meanwhile, Alice takes complete charge of transmitting dummy OOK symbols over the $f_{AB}$ band. 
Note that the helper transmits his information symbols in all the $2n$ time-slots, whereas Alice transmits her information symbols only in the first $n$ time-slots. As a result, the DTRTF scheme can be viewed as a helper-friendly scheme where the helper retains his transmission rate while Alice takes a hit on her transmission rate by 50\%. For the same reason, we refer to our strategy as the rate-three-fourth scheme since the total number of information symbols transmitted over the two frequency bands across $2n$ time-slots is three-fourth of the number of the transmitted symbols without any countermeasure. Also, we refer to our scheme as delay-tolerant as Bob needs to wait for at least $n$ time-slots to decode Alice's bits.

Next, we will explain the signalling methods over the $f_{HB}$ and $f_{AB}$ bands in detail. 

\subsection{Signalling over $f_{HB}$ Band}

As shown in Fig. \ref{fcbkn}, we ask Alice and the helper to transmit their information symbols using a portion of their energies over $f_{HB}$ during the first $n$ time-slots. During the $k^{th}$ time-slot, where $k \in [n]$, Alice communicates her bits modulated using OOK, which is scaled by a factor $\sqrt{(1-\alpha)E_H}$, while the helper communicates his information symbols scaled by a factor $\sqrt{\alpha E_H}$. Here, $\alpha\in (0,1)$ denotes the energy-splitting factor, which is a design parameter. The received symbol at Bob during the $k^{th}$ time-slot, denoted as $r_{B,k}$, is given by
\begin{eqnarray}{rcl}
\hspace{-4mm} \textcolor{black}{   r_{B,k}=\sqrt{(1-\alpha)E_H}h_{AB,k}x_k+\sqrt{\alpha E_H}y_{k} h_{HB,k}+w_{B,k},}
\end{eqnarray}
where \textcolor{black}{$h_{AB,k}\sim \mathcal{CN}(0, 1)$ is the channel between Alice and Bob that captures small-scale fading,} $x_k\in \{0, 1\}$
denotes Alice's bits, 
\textcolor{black}{$y_k \in \mathcal{S}(j+1)$, $\mathcal{S}$ denotes the complex constellation used by the helper node (this could be $M$-PSK, $M$-QAM), $j \in \{0,1, \ldots, M-1\}$ denotes the 
information symbols transmitted by the helper using a coherent modulation scheme of order $M$},
$h_{HB,k}\sim \mathcal{CN}(0, 1 +\partial)$ is the channel between the helper and Bob,  $w_{B,k}\sim \mathcal{CN}(0, N)$ is the additive white Gaussian noise (AWGN) at Bob, and the subscript $k$ denotes the $k^{th}$
time-slot. 
Since the helper is equipped with an FDR, he receives Alice's signals transmitted over $f_{HB}$, while simultaneously transmitting his information symbols. \textcolor{black}{The received vector at the helper in the $k^{th}$ time-slot, denoted using $\mathbf{r}_{H,k}$, is given by
\begin{eqnarray}{rcl}\label{rc}
\textcolor{black}{\mathbf{r}_{H,k}=\sqrt{(1-\alpha)E_H}\mathbf{h}_{AH,k}x_k+ \mathbf{h}_{HH,k}+\mathbf{w}_{H,k},}
\end{eqnarray}
where $\mathbf{h}_{AH,k}\sim \mathcal{CN}(\mathbf{0}_{N_H}, \sigma^{2}_{AH,k}\mathbf{I}_{N_H})$ is the $N_H \times 1$ channel between Alice and the helper that captures small-scale fading, $\mathbf{h}_{HH,k} \sim \mathcal{CN} (\mathbf{0}_{N_H},\alpha \rho \mathbf{I}_{N_H})$ is  $N_H \times 1$ loop interference (LI) channel \cite{LInew,LInew1,LI} at the helper, and $\mathbf{w}_{H,k}\sim \mathcal{CN}(\mathbf{0}_{N_H}, N \mathbf{I}_{N_H})$ is the AWGN at
the helper.} 
Given that Alice is the victim node, our priority is to ensure utmost reliability in transmitting her bits to Bob. 
To achieve this, the helper decodes Alice's bit transmitted during the $k^{th}$ time-slot, represented as $\bar{x}_k$, based on the received vector $\mathbf{r}_{C,k}$.
Subsequently, owing to $n$ time-slots delay offered by his practical FDR, the helper incorporates this decoded information into his information symbol through energy and phase modifications during the $(n+k)^{th}$ time-slot. If $\bar{x}_k=1$, the helper transmits his information symbol without any modifications, and the received symbol at Bob during the $(n+k)^{th}$ time-slot, denoted as $r_{B,n+k}$, is
\begin{eqnarray}{rcl}
\textcolor{black}{r_{B,n+k}=\sqrt{E_H}h_{HB,n+k}y_{n+k} + w_{B,n+k},}
\end{eqnarray}
where $h_{HB,n+k}\sim \mathcal{CN}(0, 1+\partial)$ is the channel between the helper and Bob \textcolor{black}{that captures small-scale fading}, 
\textcolor{black}{$y_{n+k} \in \mathcal{S}(j+1)$},
$w_{B,n+k}\sim \mathcal{CN}(0, N)$ is the AWGN at Bob,  and  subscript $(n+k)$ denotes the $(n+k)^{th}$ time-slot. 
If $\bar{x}_k=0$, the helper transmits  his information symbol with an energy of $2-\alpha$ and incorporates an additional phase shift of $\theta$, such that 
\begin{eqnarray}{rcl}
\textcolor{black}{r_{B,n+k}=\sqrt{(2-\alpha)E_H}h_{HB,n+k}y_{n+k} e^{
\iota \theta}+w_{B,n+k}.}
\end{eqnarray}
It should be noted that Alice remains inactive during the $(n+k)^{th}$ time-slot over $f_{CB}$ band, while her information symbols are transmitted in all the $2n$ time-slots.

\subsection{Signalling over $f_{AB}$ Band}
As Dave is constantly monitoring the statistical distribution on the energies of the received symbols on the $f_{AB}$ band, it is crucial to maintain OOK signalling in each time-slot. To achieve this, as shown in Fig. \ref{fabntimeslots}, Alice and the helper combine their remaining energies, using a preshared key over  $f_{AB}$ band, creating pseudorandom bit sequences during the first $n$ time-slots. The received symbol at Dave during the $k^{th}$ time-slot, denoted as $r_{D,k}^{AB}$, can be represented as 
\textcolor{black}{
\begin{eqnarray}{rcl}
r_{D,k}^{AB}&=&\sqrt{2E_{A} - (1 - \alpha)E_{H}}h_{AD,k}x_{d,k} \nonumber \\ &+& \sqrt{(1 - \alpha)E_{H}}h_{HD,k}x_{d,k} +w_{D,k}^{AB},
\end{eqnarray}}
where $h_{AD,k}\sim \mathcal{CN}(0, 1)$ is the channel between Alice and Dave \textcolor{black}{that captures small-scale fading}, $x_{d,k}\in \{0, 1\}$
denotes pseudorandom bits, $h_{HD,k}\sim \mathcal{CN}(0, 1+\partial)$ is the channel between the helper and Dave, and $w_{D,k}^{AB}\sim \mathcal{CN}(0, N)$ is the AWGN at Dave. During the $(n+k)^{th}$ time-slot, the helper remains silent over $f_{AB}$ band, whereas Alice transmits a pseudorandom OOK sequence from the alphabet $\{0, \sqrt{2E_A}\}$. \textcolor{black}{We highlight that the proposed strategy is feasible as long as the average energies of Alice and the helper satisfy the relation $E_A \geq \frac{E_H}{2}$, which is necessary to ensure that Alice transmits non-zero energy on  $f_{AB}$ band in $k^{th}$ time-slot.}

\subsection{Salient Features of the DTRTF Scheme}\label{SFsub}
In this subsection, we present the salient features of the DTRTF scheme with respect to the energy statistics maintained over the $f_{AB}$ and $f_{HB}$ bands, and then propose a suitable optimization problem to choose the energy-splitting factor $\alpha$. 

\begin{Proposition}\label{SF}
\textcolor{black}{For the DTRTF scheme, when $\partial = 0$},\footnote{\textcolor{black}{ Robustness analyses on the reliability and covertness for $\partial \neq 0$ are discussed in the later sections.}} 
\begin{enumerate}
    \item the statistical distribution of the energies of the received symbols over the $f_{AB}$ band is identical before and after the countermeasure, for any $\alpha \in (0, 1)$.
    
    \item The average energy of the received symbols over the $f_{CB}$ band  for $2n$ time-slots is unity, which is consistent with the case before the countermeasure, for any $\alpha \in (0, 1)$.
\end{enumerate}
\end{Proposition}

\begin{proof}
     These statements can be proved by following the signal model of the DTRTF scheme.
\end{proof}

\noindent Although the average energy of the received symbols is maintained over the $f_{HB}$ band for any $0 < \alpha < 1$, it is clear that the statistical distribution of the energies of the received symbols before and after the countermeasure are not identical owing to scaling of the information symbols as a function of $\alpha \in (0, 1)$. Similarly, the instantaneous energy of the received symbols after the countermeasure is also not the same as that before the countermeasure, and moreover it is also a function of $\alpha$. Therefore, it is imperative to characterize the probability with which the DTRTF scheme can be detected at Dave using the KLD-estimator based detector and the instantaneous energy based detector. On the other hand, with respect to the reliability feature offered by the proposed DTRTF scheme, since the information symbols of Alice and the helper are scaled as a function of $\alpha$, the average error-rates of the two users is also a function of $\alpha$. Thus, to ensure high reliability and \textcolor{black}{covertness}, we propose Problem \ref{main_opt_problem_n_time_slots}, which aims to determine the optimal value of $\alpha$, denoted as $\alpha_{opt}^{\Omega}$, which jointly minimizes the sum of $P_{Eavg}^{\Omega}$ and $P_{Davg}^{\Omega}$, where $P_{Eavg}^{\Omega}$ denotes the average probability of decoding error at Bob and $P_{Davg}^{\Omega}$ denotes the average probability of detection at Dave when using the KLD-estimator based detector and the instantaneous energy detector. 
\begin{mdframed}
\begin{problem}
\label{main_opt_problem_n_time_slots}
For a given noise variance N, we solve: $\alpha_{opt}^{\Omega}= \arg \mathop {\min }\limits_{\alpha \in (0, 1)} P_{Eavg}^{\Omega} +P_{Davg}^{\Omega}.$
\end{problem}
\end{mdframed}
In order to address Problem \ref{main_opt_problem_n_time_slots}, it is imperative to derive the expressions for both $P_{Eavg}^{\Omega}$ and $P_{Davg}^{\Omega}$ as a function of $\alpha$. \textcolor{black}{In addition to this, we would like to highlight that the expressions for $P_{Eavg}^{\Omega}$ and  $P_{Davg}^{\Omega}$ also depend on the modulation scheme used by the helper and the corresponding energies $E_A$ and $E_H$.  Henceforth, for the rest of the sections, we set $E_{A} = 0.5$ and $E_{H} = 1$, although all the results can be generalized for any $E_{A}, E_{H}$ such that $E_{A} \geq \frac{E_{H}}{2}$. Furthermore, we pick Charlie as the helper, who uses 
$M$-ary Phase Shift Keying (PSK) to communicate his information symbols to Bob.}\footnote{\textcolor{black}{For the impact of the proposed strategies when QAM is used, we refer the readers to Section \ref{Robustness}.}}
\textcolor{black}{We have chosen $M$-PSK for Charlie, owing to the fact that $M$-PSK offers the worst-case scenario in terms of covertness. 
In the next section, we discuss several decoding strategies at Bob, to arrive at an expression for $P_{Eavg}^{\Omega}$, when the DTRTF scheme is implemented for $\theta =\pi/M$.}

\begin{figure*}
\vspace{-0.5cm}
\begin{small}
\begin{eqnarray}{rcl}\label{JMAP}
\hat{a}_k, \hat{b}_k, \hat{b}_{n+k} &=& \arg \mathop {\max }\limits_{a_k,b_k,b_{n+k}} f\left( r_{B,k},r_{B,n+k} \left| x_k \right.=a_k, y_k=e^{-\frac{\iota 2 \pi b_k}{M}},  y_{n+k}=e^{-\frac{\iota 2 \pi b_{n+k}}{M}}, h_{CB,k},h_{CB,n+k} \right)
\end{eqnarray}
\begin{eqnarray}{rcl}\label{JMAPre}
\hat{a}_k, \hat{b}_k, \hat{b}_{n+k}&=&\arg \mathop {\max }\limits_{a_k,b_k,b_{n+k}} \left[ f_1\left(r_{B,k} \left| x_k \right.=a_k,y_k=e^{-\frac{\iota 2 \pi b_k}{M}}, h_{CB,k}  \right) \right.\nonumber \\ & &
\left.
\sum_{\hat{x}_k \in \{a_k, \bar{a}_k\}} P_{a_k\hat{x}_k} f_{2}\left(r_{B,n+k} \left| \bar{x}_k, y_{n+k} \right.=e^{-\frac{\iota 2 \pi b_{n+k}}{M}}, h_{CB,n+k}  \right) \right]
\end{eqnarray}
\begin{eqnarray}{rcl}\label{fx1}
f\left( r_{B,1},r_{B,2}\left| x_1 \right.=1, y_1=e^{-\frac{\iota 2 \pi b_1}{M}}, y_2=e^{-\frac{\iota 2 \pi b_2}{M}}, h_{CB,1}, h_{CB,2} \right)=f_{11} (P_{11} f_{21} + P_{10} f_{20})
\end{eqnarray}
\begin{eqnarray}{rcl}\label{fx0}
f\left( r_{B,1},r_{B,2}\left| x_1 \right.=0, y_1=e^{-\frac{\iota 2 \pi b_1}{M}}, y_2=e^{-\frac{\iota 2 \pi b_2}{M}}, h_{CB,1}, h_{CB,2} \right)=f_{10} (P_{00}  f_{20} + P_{01} f_{21})
\end{eqnarray}
\end{small}  
\hrule
\end{figure*}
\section{Error Analysis of the DTRTF Scheme} \label{DecodingatBob}

\textcolor{black}{Recall that based on the decoded bit $\bar{x}_k$, Charlie transmits his information symbol during the $(n+k)^{th}$ time-slot to Bob. As a result, the probability of decoding error of Alice's bits at Charlie will influence the probability of decoding error at Bob. In light of this, we will first discuss the error analysis at Charlie.}

\textcolor{black}{Given that Charlie is equipped with FDR, he listens to Alice's bit $x_k$ transmitted during the $k^{th}$ time-slot, while transmitting his own $M$-PSK symbol. Using the received vector at Charlie (refer to \eqref{rc}), SIC block of the FDR reduces the interference level from $\rho$ to $\rho_{th}$ thereby resulting in $\tilde{\mathbf{r}}_{H,k} = \sqrt{(1-\alpha)E_H}\mathbf{h}_{AH,k}x_k+ \tilde{\mathbf{h}}_{HH,k}+\mathbf{w}_{H,k},$ wherein $\tilde{\mathbf{h}}_{HH,k} \sim \mathcal{CN} (\mathbf{0}_{N_H},\alpha \rho_{th} \mathbf{I}_{N_H})$. Subsequently, Charlie sets an optimal threshold, represented as $\tau_{opt}$, which is computed using the various energy levels used by Alice, the second order statistics of the AWGN, and the statistics of the residual interference offered by his FDR. Subsequently, Charlie performs non-coherent energy detection by comparing $\tau_{opt}$ with the energy of the received vector, i.e., $|\tilde{\mathbf{r}}_{H,k}|^2$. The decision rule is as follows: (a) if $|\tilde{\mathbf{r}}_{H,k}|^2 > \tau_{opt}$, then $x_k=1$, and  (b) if $|\tilde{\mathbf{r}}_{H,k}|^2 < \tau_{opt}$, then $x_k=0$. For the mentioned non-coherent energy detector, the probability of decoding bit-0 as bit-1 and bit-1 as bit-0, are denoted using $P_{01}$ and $P_{10}$, respectively. We highlight that $P_{01}$ and $P_{10}$ inherently takes into account the residual interference of the FDR. Given that the probability of decoding error at Bob depends on the probability of decoding error at Charlie, we perform the error analysis at Bob by assuming that Bob have complete knowledge about $P_{01}$ and $P_{10}$.}

Now, we will discuss a practical decoding strategy at Bob to retrieve both Alice's  and Charlie's symbols transmitted on the $f_{CB}$ band.
Considering that Alice's information is transmitted across both $k^{th}$ and $(n+k)^{th}$ time-slot, Bob can employ an optimal joint decoder to recover $x_k$, $y_k$, and $y_{n+k}$ from the received symbols $r_{B,k}$ and $r_{B,n+k}$. 
In such a case, the decoding metric of the Joint Maximum A Posteriori (JMAP) decoder is given by \eqref{JMAP}, 
where $f\left( r_{B,k},r_{B,n+k} | x_k, y_k,  y_{n+k}, h_{CB,k},h_{CB,n+k} \right)$ is the conditional probability density function (CPDF) of $r_{B,k}$ and $r_{B,n+k}$, given $x_k$, $y_k$, $y_{n+k}$, $h_{CB,k}$ and $h_{CB,n+k}$. In \eqref{JMAP}, $a_k\in \{0,1\}$, while $b_k,b_{n+k}\in\{0, 1,..., M-1\}$, thereby representing the search space of the JMAP decoder.
Given the statistical independence of the symbols $r_{B,k}$ and $r_{B,n+k}$ conditioned on $x_k$, $\bar{x}_k$, $y_k$, $y_{n+k}$, $h_{CB,k}$ and $h_{CB,n+k}$, the JMAP decoder can be reformulated as \eqref{JMAPre},
where $f_1\left(r_{B,k} \left| x_k \right.=a_k, y_k=e^{-\frac{\iota 2 \pi b_k}{M}}, h_{CB,k} \right)$ is the CPDF of $r_{B,k}$ given $x_k$, $y_k$ and $h_{CB,k}$,  and 
$f_2\left(r_{B,n+k} \left| \bar{x}_k, y_{n+k} \right.=e^{-\frac{\iota 2 \pi b_{n+k}}{M}}, h_{CB,n+k}  \right)$ is the CPDF of $r_{B,n+k}$ given $\bar{x}_k$, $y_{n+k}$, $h_{CB,n+k}$, and $\bar{a}_k$ represents the complement of $a_k$. 
Henceforth, for exposition, without the loss of generality, we assume $k=1$ and $n=1$ for the ease of notations in our error analysis.
Consequently, ${r}_{B,1}$
is received on time-slot 1 and ${r}_{B,2}$ is received on time-slot 2.
The CPDF of ${r}_{B,1}$ given $x_1 = 1$, denoted as $f_{11}\left({r}_{B,1}|x_1 = 1,y_1=e^{-\frac{\iota 2 \pi b_1}{M}}, h_{CB,1}\right)$, and the CPDF of ${r}_{B,1}$ given $x_1 = 0$, denoted as $f_{10}\left({r}_{B,1}|x_1 = 0,y_1=e^{-\frac{\iota 2 \pi b _1}{M}}, h_{CB,1}\right)$, based on the transmitted bit $x_1$, take the complex Gaussian structures with variance $N_{1b}=N+1-\alpha$ and $N_{0b}=N$, respectively. Similarly, the CPDF of $r_{B,2}$ given $\bar{x}_1$, $y_2$, and $h_{CB,2}$, denoted as $f_{21}\left(r_{B,2} \left| \bar{x}_1=1, y_2 \right.=e^{-\frac{\iota 2 \pi b_2}{M}}, h_{CB,2} \right)$ for $\bar{x}_1=1$, and $f_{20}\left(r_{B,2} \left| \bar{x}_1=0, y_2 \right.=e^{-\frac{\iota 2 \pi b_2}{M}}, h_{CB,2} \right)$ for $\bar{x}_1=0$, can be determined based on $\bar{x}_1$, and take the complex Gaussian structure with variance $N_{0b}=N$, however, with different mean values based on $\bar{x}_1$. The overall CPDF for $x_1=1$ and $x_1=0$ 
can be expressed as given in \eqref{fx1} and \eqref{fx0}, respectively.
In the above equations, $P_{ij}$ presents the probability that Charlie decodes Alice's bit-$i$ as bit-$j$, where $i, j \in \{0, 1\}$ \cite{V3,V4}.
Note that \eqref{fx1} and \eqref{fx0} composes of Gaussian mixtures that are scaled by probabilities of decoding error introduced by Charlie. Due to the complexities associated with handling Gaussian mixtures, calculating the overall probability of decoding error using the JMAP decoder defined in \eqref{JMAPre} becomes challenging.
To address this, we propose a sub-optimal decoder known as Sub-Optimal DTRTF (SODTRTF) decoder, where Bob decodes Charlie's first symbol in time-slot 1, and then jointly decodes Alice's bit and Charlie's second symbol in time-slot 2. To solely decode Charlie's transmitted symbol $y_1$ during time-slot 1, Bob obtains $\hat{a}_1$, $\hat{b}_1$ given by 

\vspace{-0.3cm}
\begin{small}
\begin{eqnarray}{rcl}\label{Dtimeslot1}
\hspace{-5mm}\hat{a}_1,\hat{b}_1&=& \arg \mathop {\max }\limits_{a_1,b_1} f_1\left(r_{B,1} \left| {x}_1=a_1, y_1 \right.=e^{-\frac{\iota 2 \pi b_1}{M}}, h_{CB,1}  \right),
\end{eqnarray} 
\end{small}

\noindent and discards $\hat{a}_1$.
Subsequently, during time-slot 2, Bob jointly decodes Alice's bit $x_1$ and Charlie's symbol $y_2$ as 

\vspace{-0.3cm}
\begin{small}
\begin{eqnarray}{rcl}\label{D21}
\hspace{-5mm}\hat{a}_1, \hat{b}_2 &=& \arg \mathop {\max }\limits_{a_1,b_2} f_2\left(r_{B,2} \left| {x}_1=a_1, y_2 \right.=e^{-\frac{\iota 2 \pi b_2}{M}}, h_{CB,2}  \right).
\end{eqnarray}
\end{small}

\noindent In contrast to the JMAP decoder, we highlight that the decoding metrics in \eqref{Dtimeslot1} and \eqref{D21} are analytically tractable for analyzing their error performance. 
For a given $h_{CB,1}$ and $h_{CB,2}$, if $P_{E1}^{\Omega}$ and $P_{E2}^{\Omega}$ represent the probability of decoding error associated with time-slot 1 and time-slot 2, respectively, then the overall probability of decoding error of the SODTRTF decoder, denoted as $P_{E}^{\Omega sub}$, can be upper bounded as
\begin{eqnarray}{rcl}\label{PesubD}
P_{E}^{\Omega sub} \leq P_{E1}^{\Omega}+P_{E2}^{\Omega}.
\end{eqnarray}

\noindent Next, we provide upper bounds on $P_{E1}^{\Omega}$ and $P_{E2}^{\Omega}$.

\subsection{Upper Bound on $P_{E1}^{\Omega}$} 
With $a_1,b_1$ representing the indices jointly chosen by Alice and Charlie, the decoding metric in \eqref{Dtimeslot1} introduces an error event if $(b_1 \neq \hat{b}_1)$ and $(a_1=\hat{a}_1) \cup (a_1 \neq \hat{a}_1)$.
If the probability of this error event is denoted as $\Pr \left(\hat{b}_1\neq b_1\right)$, then the overall probability of decoding error for $y_1$, conditioned on $h_{CB,1}$, is given by
\begin{eqnarray}{rcl}\label{Pesumsum}
\hspace{-5mm}P_{E1}^{\Omega} = \frac{1}{2M} \sum_{ b_1 \in \{0, 1, \ldots, M-1\}} \Pr\left(\hat{b}_1\neq b_1\right).
\end{eqnarray} 
Prior to calculating every term of the right-hand-side of \eqref{Pesumsum}, in Fig. \ref{Constellation_diagram_TS1}, we present a 2-dimensional  constellation diagram depicting the superposition of the information symbols of Alice and Charlie during  time-slot 1 for the case when $M=4$. 
 The dots and circles within the diagram correspond to the $M$-PSK constellation, for $x_1=0$ and $x_1=1$, respectively. In particular, information symbols $(a_1, b_1)$ of the form $(0, b_1)$ are represented as dots, whereas the information symbols $(a_1, b_1)$ of the form $(1, b_1)$ are represented as circles. Moving forward, we will refer to Alice's bits and Charlie's symbols interchangeably as $(a_1, b_1) \in \{0, 1\} \times \{0, 1, \ldots, M-1\}$ or as points within the constellation diagram, as these two representations share a one-to-one mapping.

In the next theorem, we use this one-to-one mapping on the representations to present an upper bound on $
P_{E1}^{\Omega}$.

\begin{figure}[ht!]%
\begin{center}
        \subfloat[]
{\includegraphics[scale=0.27]{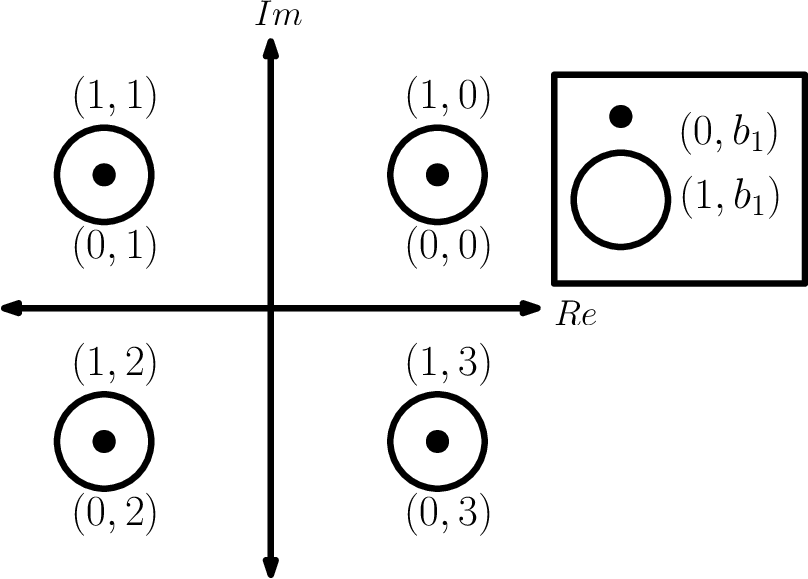}\label{Constellation_diagram_TS1}}
    \hfil
    \subfloat[]
{\includegraphics[scale=0.27]{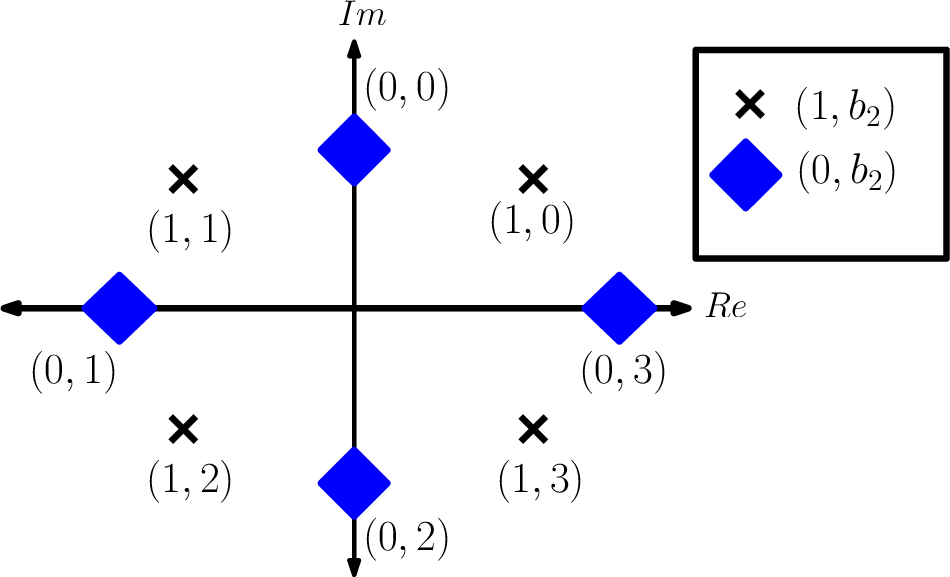}\label{Constellation_diagram_TS2}}
    
\end{center}
    \caption{\textcolor{black}{(a) Constellation diagram illustrating the
    symbols jointly communicated by Alice and Charlie during time-slot 1, presented in the form of $(x_1, y_1)$. (b) 
    Constellation diagram illustrating the symbols communicated by Charlie during time-slot 2, presented in the form of $(\bar{x}_1, y_2)$.}}
    \label{Constdiagram}
    \vspace{-0.3cm}
\end{figure}
\begin{theorem}\label{T1}
At high SNR i.e.,  $N \ll 1$, an upper bound on $P_{E1}^{\Omega}$ conditioned on $h_{CB,1}$, denoted as $P_{E1}^{\Omega ub}$, is given by 
\begin{eqnarray}{rcl}\label{T1eqn}
 P_{E1}^{\Omega ub}&=&\Pr[(1,1)\rightarrow(0,0)]+\Pr[(1,1)\rightarrow(1,0)] \nonumber \\&+&\Pr[(0,1)\rightarrow(0,0)]+\Pr[(0,1)\rightarrow(1,0)],   
\end{eqnarray}
where $\Pr[(a_1,b_1)\rightarrow(\bar{a}_1,\bar{b}_1 )]$ denotes the probability that $(a_1,b_1)$ is decoded as $(\bar{a}_1,\bar{b}_1)$, when $(a_1,b_1)$ is transmitted by Alice and Charlie.
\end{theorem}

\begin{figure*}[ht!]
\vspace{-0.3cm}
\textcolor{black}{
\begin{eqnarray}{rcl}\label{Prop2exp}
\hspace{-5mm}  P_{A1}=1-Q_{1}\left(\sqrt{N_{1b} r},\sqrt{Ns}\right),  
P_{A2}=Q\left(\sqrt{\frac{\alpha}{N_{1b}}},|h_{CB,1}|\right), 
P_{A3}=Q\left(\sqrt{\frac{\alpha}{N}},|h_{CB,1}|\right), 
P_{A4}=Q_{1}\left(\sqrt{Nr},\sqrt{N_{1b}s}\right) 
\end{eqnarray}
   \begin{eqnarray}{rcl} \label{TH3exp1}
& & P_{A1}^{avg}=\left(\frac{N_{1b} \sqrt{\alpha}}{(1-\alpha)^2}+1\right)^{(-1)},
P_{A2}^{avg}=\sum_{i=1}^{3}\frac{k_i}{\frac{t_i \alpha}{N_{1b}}+1},
P_{A3}^{avg} =\sum_{i=1}^{3}\frac{k_i}{\frac{t_i \alpha}{N}+1},
P_{A4}^{avg}=\left(\frac{N}{N_{1b}}\right)^{\frac{N_{1b}}{1-\alpha}}P_{A1}^{avg}
\end{eqnarray} 
   \begin{eqnarray}{rcl} \label{TH3exp2}
P_{B}^{avg}=\sum_{i=1}^3 \frac{k_{i}}{\frac{t_{i}d}{2N}+1},
P_{B2}^{avg}=\sum_{i=1}^3 \frac{k_{i}}{\frac{t_{i}}{2N}+1}
\end{eqnarray} }
\hrule
\end{figure*}

Henceforth, for the ease of notations, we denote
 $\Pr[(1, 1) \rightarrow(0, 0)]$, $\Pr[(1,1)\rightarrow(1,0)]$, $\Pr[(0,1)\rightarrow(0,0)]$ and $\Pr[(0,1)\rightarrow(1,0)]$ using $P_{A1}$, $P_{A2}$, $P_{A3}$ and $P_{A4}$, respectively. 
In the following proposition, we provide the expressions for $P_{A1}$,  $P_{A2}$, $P_{A3}$ and $P_{A4}$, which can be derived from their definitions using first principles.

\begin{Proposition} \label{PA}
For a given $\alpha$, $N$, and $h_{CB,1}$, the expressions of $P_{A1}$, $P_{A2}$, $P_{A3}$ and $P_{A4}$ are given by \eqref{Prop2exp}, where $r=\frac{2\sqrt{\alpha} |h_{CB,1}|^2}{(1-\alpha)^2}$,  $s=\frac{2}{1-\alpha}\left(ln \frac{N_{1b}}{N}+ \frac{\sqrt{\alpha} |h_{CB,1}|^2}{1-\alpha} \right)$, and $Q_1(.,.)$ and $Q(.)$ denote Marcum-Q function of first order and Q-function, respectively.

\end{Proposition}

Now that we have an upper bound on $P_{E1}^{\Omega}$, the next task is to derive an upper bound on $P_{E2}^{\Omega}$.

\subsection{Upper Bound on $P_{E2}^{\Omega}$}
With $a_1,b_2$ representing the indices jointly chosen by Alice and Charlie, the decoding metric in \eqref{D21} raises an error event, denoted by $\Delta^{(1)}_{(a_1,b_2)\rightarrow (\hat{a}_1,\hat{b}_2)}$, is given by, 

\begin{small}
   \begin{eqnarray}{rcl} \label{delta2}
\Delta^{(1)}_{(a_1,b_2)\rightarrow (\hat{a}_1,\hat{b}_2)} \triangleq \frac{f_{2}\left(r_{B,2} \left | {x_1}=a_1, y_2 \right.=e^{-\frac{\iota 2 \pi b_2}{M}}, h_{CB,2}  \right)}{f_{2}\left(r_{B,2} \left | {x_1}=\hat{a}_1, y_2 \right.=e^{-\frac{\iota 2 \pi \hat{b}_2}{M}}, h_{CB,2}  \right)} \leq 1, \nonumber \\
\end{eqnarray}
\end{small}

\noindent when $(\hat{a}_1 \neq a_1)$ or $(\hat{b}_2 \neq b_2)$.
Using $\Delta^{(1)}_{(a_1,b_2)\rightarrow (\hat{a}_1,\hat{b}_2)}$, the overall probability of decoding error for $x_1$ and $y_2$, conditioned on $h_{CB,2}$ is 
 \begin{eqnarray}{rcl}\label{TS2sum}
 P_{E2}^{\Omega} = \frac{1}{2M} \sum_{a_1=0}^{1} \sum_{b_2=0}^{M-1} \Pr\left(\Delta^{(1)}_{(a_1,b_2)\rightarrow (\hat{a}_1,\hat{b}_2)}\right).
\end{eqnarray}  
\noindent Prior to analyzing every term in the right-hand-side of \eqref{TS2sum}, in Fig. \ref{Constellation_diagram_TS2}, we visualize a 2-dimensional constellation diagram representing Charlie's symbols transmitted during time-slot 2 for $M=4$.
In Fig. \ref{Constellation_diagram_TS2}, the black crosses and the blue diamonds represent the $M$-PSK constellation, for   $\bar{x}_1=1$ and $\bar{x}_1=0$, respectively.
As we proceed, we will refer to Alice's and Charlie's symbols as $(a_1, b_2) \in \{0, 1\} \times \{0, 1, \ldots, M-1\}$, or as points within the constellation, as these two share a one-to-one mapping. 

In the next theorem, we present an upper bound on $
P_{E2}^{\Omega}$ by using the one-to-one mapping on the representations.

\begin{theorem}\label{Pe2th}
At high SNR i.e., $N \ll 1$, an upper bound on $P_{E2}^{\Omega}$ conditioned on $h_{CB,2}$, denoted as $P_{E2}^{\Omega ub}$, can be expressed as 
\textcolor{black}{
\begin{eqnarray}{rcl}
P_{E2}^{\Omega ub}&=&\frac{1}{2}[ 2 P_{11} (P_{B1}+P_{B2})+ P_{01} (1-P_{B1}) \nonumber\\ &+&  2 P_{00} P_{B3}+ P_{10} (1-P_{B3})], 
\end{eqnarray}}
where $P_{B1}=\Pr \left(\Delta^{(1)}_{(1, 1)\rightarrow (0, 0)}\right)$, $P_{B2}=\Pr \left(\Delta^{(1)}_{(1, 1)\rightarrow (1, 2)}\right)$ and $P_{B3}=\Pr \left(\Delta^{(1)}_{(0, 0)\rightarrow (1, 1)}\right)$.
\end{theorem}

In the following proposition, we present expressions for $P_{B1}$, $P_{B2}$ and $P_{B3}$, which can be derived from their definitions using first principles.

\begin{Proposition} \label{PB0PB1}
For a given $\alpha$, $N$ and $h_{CB,2}$, the expressions of $P_{B1}$, $P_{B2}$ and $P_{B3}$ are given by
\textcolor{black}{
\begin{eqnarray}{rcl}
\hspace{-4mm}P_{B1}=P_{B2}=Q\left(\frac{|h_{CB,2}|d}{\sqrt{2 N}}\right),
P_{B3}=Q\left(\frac{|h_{CB,2}|d}{\sqrt{2 N}}\right),
\end{eqnarray}}
where $d=\sqrt{3-\alpha-2\sqrt{2-\alpha} cos (\pi/M)}$.
\end{Proposition}

Note that we have an upper bound on $P_{E1}^{\Omega}$ and $P_{E2}^{\Omega}$, as denoted by $P_{E1}^{\Omega ub}$ and $P_{E2}^{\Omega ub}$, in Theorem \ref{T1} and Theorem \ref{Pe2th}, respectively. Given that $P_{E1}^{\Omega ub}$ and $P_{E2}^{\Omega ub}$ are functions of $h_{CB,1}$ and $h_{CB,2}$, respectively, the next task is to find expectation $\mathbb{E}_{|h_{CB,1}|^2,|h_{CB,2}|^2}[P_{E}^{\Omega sub}]$. Using the linearity property of the expectation operator, we have Theorem \ref{PEtotalDTRTFtheorem}, wherein, the upper bound is derived by using \cite{Marcum,q_approx}.
\begin{theorem} \label{PEtotalDTRTFtheorem}
\textcolor{black}{
    An upper bound on average probability of decoding error at Bob using SODTRTF decoder is given by}
\begin{eqnarray}{rcl}\label{PEtotalDTRTF}
  \hspace{-4mm} \textcolor{black}{ \mathbb{E}_{|h_{CB,1}|^2,|h_{CB,2}|^2}[P_{E}^{\Omega sub}] \leq
    P_{UEavg}^{\Omega sub}\triangleq P_{avg1}^{\Omega ube}+P_{avg2}^{\Omega ube}, }
\end{eqnarray}
\textcolor{black}{where} 
\begin{eqnarray}{rcl}
\textcolor{black}{P_{avg1}^{\Omega ube}=P_{A1}^{avg}+P_{A2}^{avg}+P_{A3}^{avg}+P_{A4}^{avg},}
\end{eqnarray}
\begin{eqnarray}{rcl}
\textcolor{black}{P_{avg2}^{\Omega ube}}&\textcolor{black}{=}&\textcolor{black}{\frac{1}{2}[ 2 P_{11} (P_{B}^{avg}+P_{B2}^{avg})+ P_{01} (1-P_{B}^{avg})}\nonumber  \\ 
 &\textcolor{black}{+}& \textcolor{black}{ 2 P_{00} P_{B}^{avg}+ P_{10} (1-P_{B}^{avg})] }.
\end{eqnarray}
The expression of $P_{A1}^{avg}$,  $P_{A2}^{avg}$, $P_{A3}^{avg}$, $P_{A4}^{avg}$, $P_{B}^{avg}$ and $P_{B2}^{avg}$ are given in \eqref{TH3exp1} and \eqref{TH3exp2}.
Also, $k_{1} = 0.168$, $k_{2} = 0.144$, $k_{3} = 0.002$, $t_{1} = 0.876$, $t_{2} = 0.525$, and $t_{3} = 0.603$.
Note that the expressions of $P_{B1}$ and $P_{B3}$ given in Proposition \ref{PB0PB1} are identical, thus averaging them over $|h_{CB,2}|$ yields the same result, denoted as $P_{B}^{avg}$.
\end{theorem}

To validate the accuracy of our derived expressions, in Fig. \ref{DTRTF_SNR}, we use monte-\textcolor{black}{carlo} simulations to plot the average probability of decoding error 
using the JMAP decoder (defined in \eqref{JMAP}, and the corresponding probability of error denoted as $P_{Eavg}^{\Omega}$) and SODTRTF decoder (defined in \eqref{Dtimeslot1} and \eqref{D21}, and the corresponding probability of error denoted as $P_{E}^{\Omega sub}$) for $N_C=2$.
Also, we plot the derived mathematical upper bound on $P_{E}^{\Omega sub}$, denoted as $P_{UEavg}^{\Omega sub}$ (given in Theorem \ref{PEtotalDTRTFtheorem}), denoted by $P_{UEavg}^{\Omega sub}$.
From the plots, it is clear that the curves of $P_{Eavg}^{\Omega}$ and $P_{E}^{\Omega sub}$ nearly overlap, demonstrating the near-optimality of  our proposed SODTRTF decoder compared to the optimal decoder.
Furthermore, the curve of  $P_{UEavg}^{\Omega sub}$ sits on the top of other two curves, which shows that $P_{UEavg}^{\Omega sub}$ is truly an upper bound. All three plots exhibit the same trend, with their global minima occurring at similar values of $\alpha$. 

\begin{figure}[ht!]
\vspace{-0.2cm}
\includegraphics[scale=0.23]{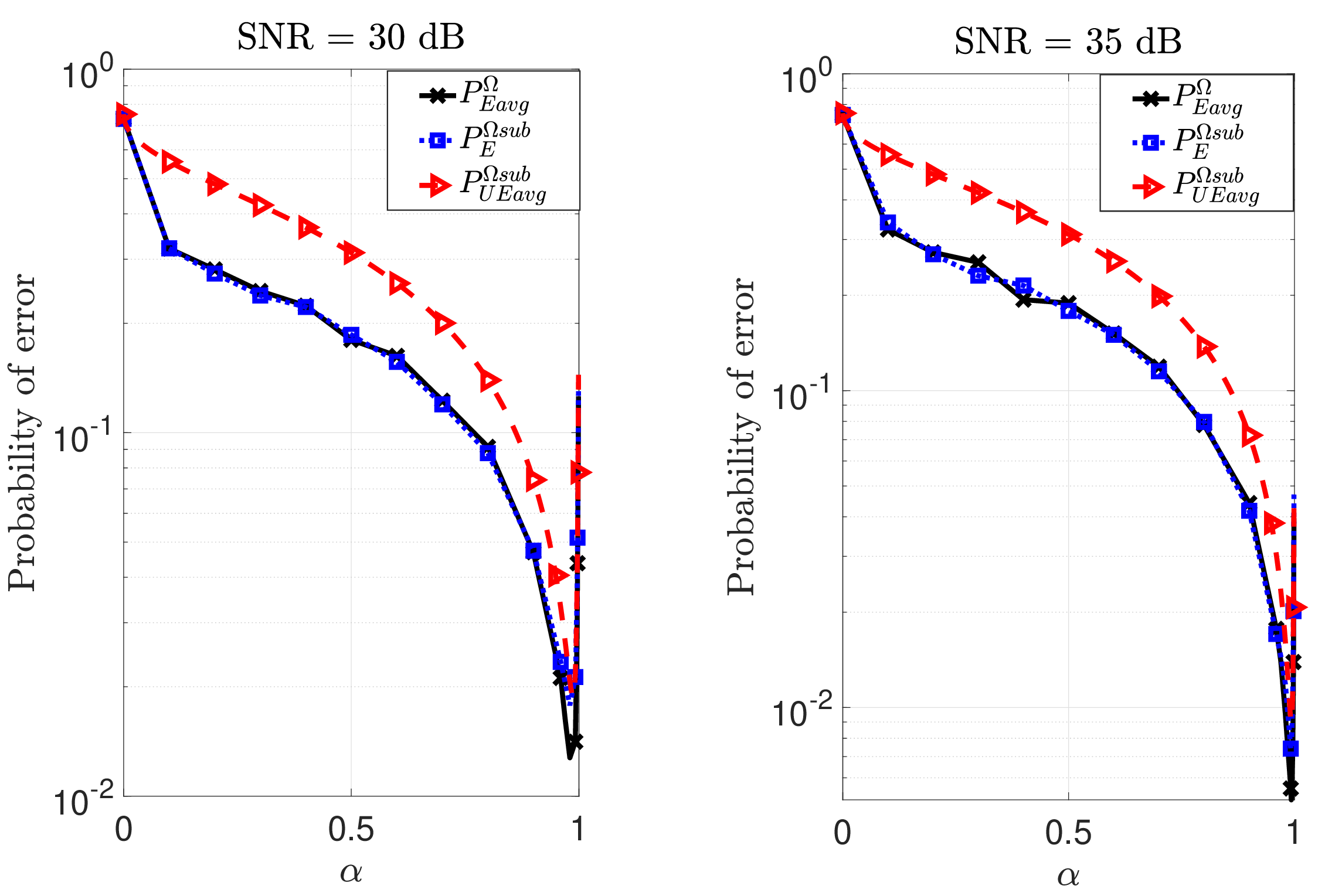}
\vspace{-0.6cm}
\caption{\textcolor{black}{Figure illustrating the error performance  of the optimal decoder ($P_{Eavg}^{\Omega}$) and the SODTRTF decoder ($P_{E}^{\Omega sub}$), as a function of the energy-splitting factor $\alpha$. Additionally, we plot the derived upper bound on the average probability of decoding error ($P_{UEavg}^{\Omega sub}$) of the SODTRTF decoder.}} 
\label{DTRTF_SNR}
\end{figure}

\textcolor{black}{Although we have proposed the DTRTF scheme for $m \geq n_{fr}$ and $m \neq 0$, we highlight that it can also be used for $m < n_{fr}$
and $m \neq 0$. However, in such a case, the residual interference would be higher than $\rho_{th}$, which in-turn results in increased values of $P_{10}$, $P_{01}$, thereby degrading the overall error probability at Bob. Despite the benefits of the DTRTF scheme in the regime $m  \geq n_{fr}$, it is evidently not applicable when $m=0$. Therefore, we present the LLCRTF scheme, which is applicable for $m<n_{fr}$ in general, and particularly relevant when $m=0$.}

\begin{figure}[ht!]
\begin{center}
\includegraphics[scale=0.35]{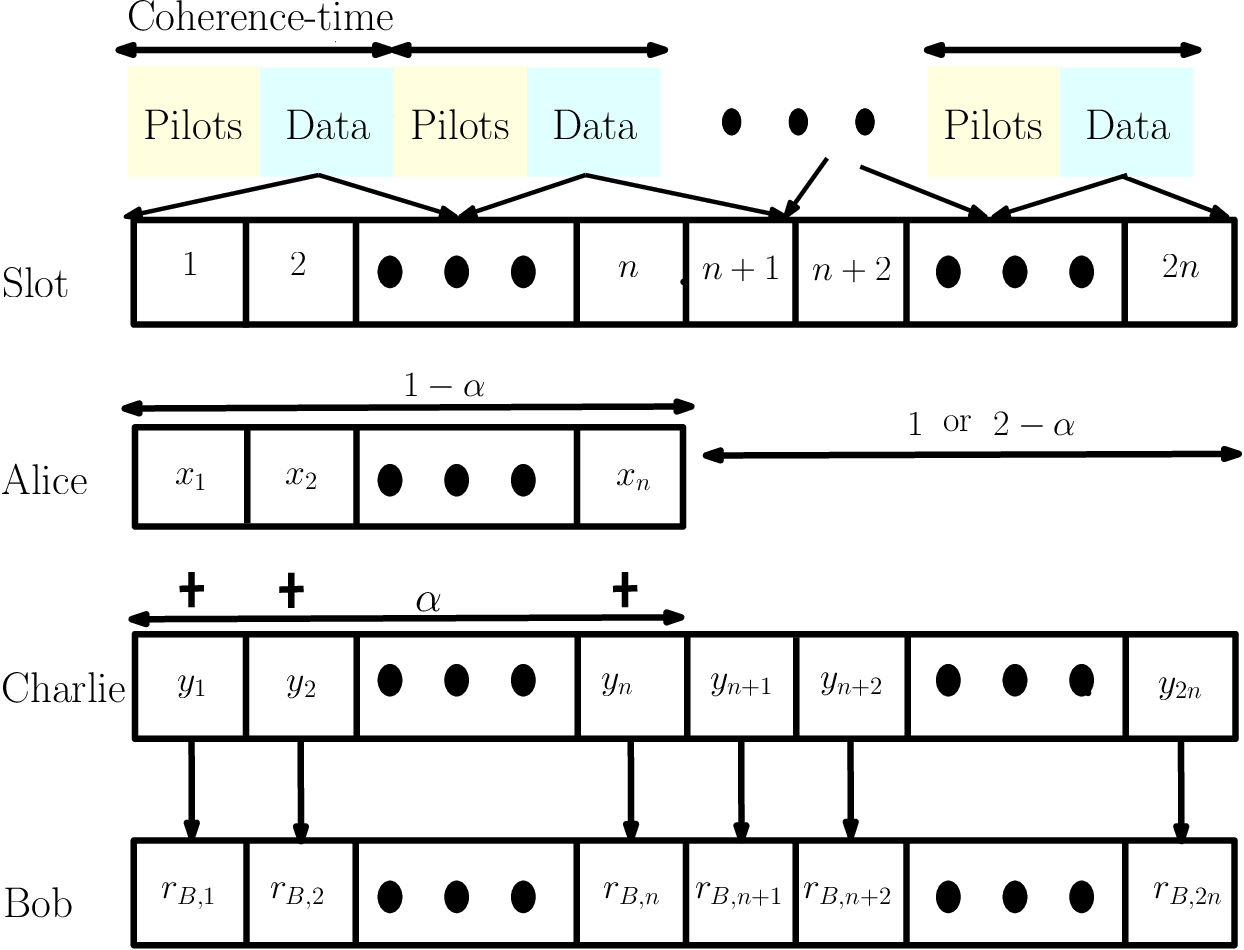}
\end{center}
\caption{\textcolor{black}{Schematic representation of the transmission of information symbols over $f_{CB}$ band in the LLCRTF scheme. 
}}
\label{nlatencytimeslots}
\end{figure}

\section{LLCRTF Mitigation Strategy} \label{LLCA}
This section presents the LLCRTF scheme, which is applicable when $m<n_{fr}$.
As shown in Fig. \ref{nlatencytimeslots},  LLCRTF scheme has the same structure as that of  DTRTF scheme. \textcolor{black}{However, in this scheme, Charlie does not embed Alice's bits into his transmitted symbols in the $(n+k)^{th}$ time-slot, owing to the fact that the latency-constraint is less than the delay offered by Charlie's FDR.} As a result, during the $(n+k)^{th}$ time-slot, Charlie selects one of the two sets of $M$-PSK symbols in a random fashion with uniform probability. Note that this modification is introduced to maintain the average energy over $f_{CB}$ band.
This randomness in the choice of the constellation is captured using a binary random variable $x_k^d \in \{0, 1\}$, such that if $x_k^d=1$, Charlie transmits his regular $M$-PSK symbol with unit energy.
On the other hand, if $x_k^d=0$, Charlie transmits his $M$-PSK symbol with $(2-\alpha)$ energy along with an additional phase-shift of $\pi/M$.
Since Charlie does not incorporate Alice's bits into his transmitted symbols, the value of $n \geq 1 $ can be chosen in a predetermined fashion by Alice, Charlie and Bob.\footnote{\textcolor{black}{Since LLCRTF does not require FDRs, this countermeasure is an appealing choice for low-powered internet-of-things networks.}} Except for the above mentioned changes on the $f_{CB}$ band, the LLCRTF and the DTRTF schemes have the same signalling method on the $f_{AB}$ band. Therefore, the salient features of  DTRTF scheme given in Proposition \ref{SF} are also applicable for LLCRTF scheme.

Compared to the DTRTF scheme, this scheme offers reduced complexity at Charlie as it eliminates the need for using an FDR to decode Alice's bits.\footnote{\textcolor{black}{As the LLCRTF scheme does not require a helper with an FDR, to implement this scheme, any of the three UEs, i.e. Tom, Frank and Charlie can be selected as the helper. However, to have uniformity with the DTRTF scheme, we will continue to discuss all the mathematical analyses for the case when Charlie acts as the helper without using his FDR capability. } } 
Since Charlie does not embed Alice's bits into his symbols, Bob needs to decode Alice's and Charlie's symbols in the $k^{th}$ time-slot, whereas decode only Charlie's symbols in the $(n+k)^{th}$ time-slot.
Similar to the DTRTF scheme, the average probability of error will be a function of $\alpha$, therefore 
 we present Problem \ref{main_opt_problem_n_time_slots_latency} that determines the optimal value of $\alpha$, denoted by $\alpha ^{\mu}_{opt}$, which jointly minimizes the average probability of decoding error at Bob  associated with decoding $x_k$, $y_k$, and $y_{n+k}$, (represented using $P_{Eavg}^{\mu}$) and the average probability of detection at Dave when using the KLD-estimator based detector and
the instantaneous energy detector (represented using $P_{Davg}^{\mu}$). 

\begin{mdframed}
\begin{problem}
\label{main_opt_problem_n_time_slots_latency}
For a given noise variance N, solve $\alpha ^{\mu}_{opt}= \arg \mathop {\min }\limits_{\alpha \in (0, 1)} P_{Eavg}^{\mu} +P_{Davg}^{\mu}.$
\end{problem}
\end{mdframed} 

We highlight that using $r_{B,k}$, Bob jointly decodes Alice's and Charlie's symbols transmitted during $k^{th}$ time-slot, and  using $r_{B,2}$, Bob decodes Charlie's symbols transmitted during $(n+k)^{th}$ time-slot. We follow the footsteps of Theorem \ref{PEtotalDTRTFtheorem}, and present an upper bound on $P_{Eavg}^{\mu}$ in Theorem \ref{PEtotalLLCRTFtheorem}.

\begin{theorem} \label{PEtotalLLCRTFtheorem}
     An upper bound on average probability of decoding error at Bob for the LCLRTF scheme is given by
\begin{eqnarray}{rcl}\label{PEtotalLLCRTF}
  \hspace{-4mm}  P_{Eavg}^{\mu} \leq
P_{UEavg}^{\mu}\triangleq P_{avg1}^{\mu ube}+P_{avg2}^{\mu ube}, 
\end{eqnarray}
where 
\textcolor{black}{
 \begin{eqnarray}{rcl}
\hspace{-6mm}P_{avg1}^{\mu ube}= \frac{2(P_{A1}^{avg}+P_{A2}^{avg}+P_{A3}^{avg}+P_{A4}^{avg})+P_{A5}+P_{A6}}{2},
\end{eqnarray}
 \begin{eqnarray}{rcl}
P_{avg2}^{\mu ube}=P_{B}^{avg}+P_{B2}^{avg}.
\end{eqnarray}}
\noindent Here,
\textcolor{black}{
\begin{eqnarray}{rcl}
 P_{A5}& \triangleq & \Pr[(1,1)\rightarrow (0,1)] =1-\left(\frac{N_{1b}}{N_{0b}}^{\frac{N_{0b}}{\alpha-1}}\right),\\
P_{A6} &\triangleq &\Pr[(0,1)\rightarrow (1,1)]=\left(\frac{N_{0b}}{N_{1b}}^{\frac{N_{1b}}{1-\alpha}}\right).
 \end{eqnarray}}
Also, the expressions of $P_{A1}^{avg}$, $P_{A2}^{avg}$, $P_{A3}^{avg}$, $P_{A4}^{avg}$, $P_{B}$ and $P_{B}$ are given in Theorem \ref{PEtotalDTRTFtheorem}. 
\end{theorem}

Now, to prove the validity of our results, in Fig. \ref{LLCRTF_SNR}, we plot the average probability of decoding error (denoted as $P_{Eavg}^{\mu}$), and the derived mathematical upper bound on $P_{Eavg}^{\mu}$, denoted as $P_{UEavg}^{\mu}$ (given in Theorem \ref{PEtotalLLCRTFtheorem}).  
From the plots, we observe that both curves exhibit similar trend, and their minima are around the same value of $\alpha$.

\begin{figure}[ht!]
\vspace{-0.2cm}
\includegraphics[scale=0.23]{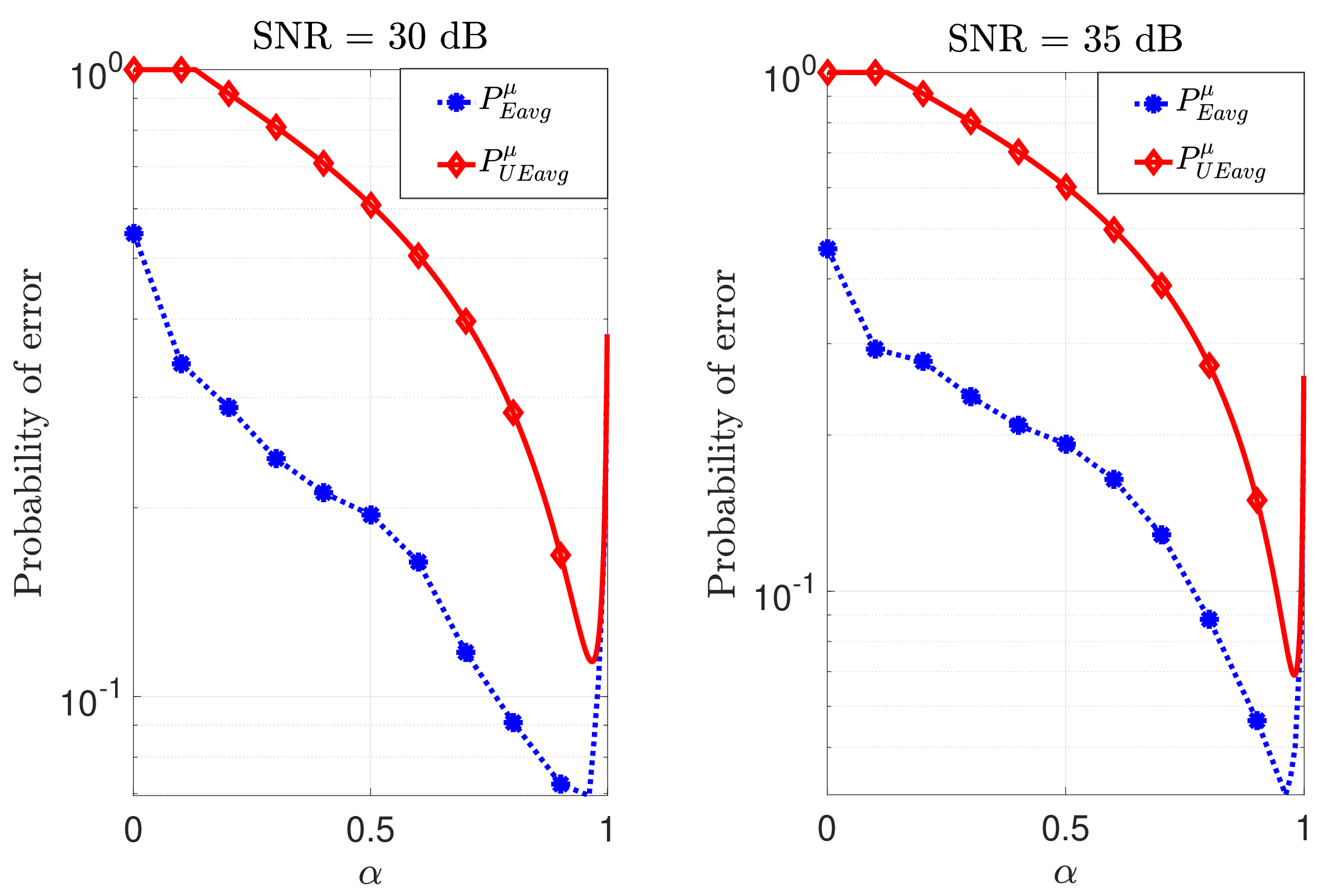}
\caption{\textcolor{black}{Figure illustrates the overall error performance of the LLCRTF scheme, denoted as $P_{Eavg}^{\mu}$, and the derived upper bound on $P_{Eavg}^{\mu}$, denoted as $P_{UEavg}^{\mu}$.}} 
\label{LLCRTF_SNR}
\end{figure}

In the following section, we will discuss the accuracy with which our proposed schemes gets detected by Dave.

\section{\textcolor{black}{Covertness} Analysis at Dave} \label{CD}
In this work, a countermeasure is said to be covert against a particular detector at Dave if it fails to detect the countermeasure with a high probability. 
Based on the threat model in Section \ref{SM}, we need to analyse the \textcolor{black}{covertness} of our countermeasures against the KLD-estimator based detector and the instantaneous energy detector on every band. Henceforth, we will discuss \textcolor{black}{covertness} analysis only on the $f_{AB}$ and the $f_{CB}$ bands since the other frequency bands are not altered implicitly. With respect to the $f_{AB}$ band, while the first result of Proposition \ref{SF} shows the effectiveness of the proposed schemes against the KLD-estimator based detector, the instantaneous energy detector is ineffective due to the non-coherent channel between Alice and Bob. Thus, we focus on the \textcolor{black}{covertness} analysis of our countermeasures when Dave uses the two detectors on the $f_{CB}$ band. However, towards solving the optimization problems posed in Problem \ref{main_opt_problem_n_time_slots} and Problem \ref{main_opt_problem_n_time_slots_latency}, we believe that deriving the corresponding expressions for $P_{Davg}^{\Omega}$ and $P_{Davg}^{\mu}$ are challenging due to the intractability in analytically characterising the probability of detection of the KLD-estimator based detector. As a consequence, in the next section, we derive the expressions for average probability of detection when using the instantaneous energy detector, and then use it in place of $P_{Davg}^{\Omega}$ and $P_{Davg}^{\mu}$ for optimizing the value of $\alpha$. Subsequently, in Section \ref{sec:kld_fcb}, we apply the solutions of these optimization problems to show that the estimates from the KLD estimator are closer to those when computing without the countermeasure.

\subsection{Instantaneous Energy Detector on $f_{CB}$ Band}

Recall that Dave has complete knowledge of the modulation schemes used by different users of the network. As a result, Dave uses the knowledge of the $M$-PSK constellation at Charlie, to monitor the envelope of the symbols on the $f_{CB}$ band using an instantaneous energy detector. Furthermore, to realize this detector, we consider a worst-case scenario for Alice wherein Dave estimates $|h_{CD,k}|$, i.e., instantaneous magnitude of the channel between himself and Charlie, in order to monitor the envelope. Given that unit energy $M$-PSK symbols are transmitted before implementing our proposed countermeasures, Dave raises a detection event on $f_{CB}$ band when the received instantaneous energy, after normalized by $|h_{CD,k}|$, is either below $1-\delta$ or above $1+\delta$, for some chosen $\delta >0$ 
\cite[Proposition 3]{ISIT}. In this context, $\delta$ is chosen such that the probability of false alarm is bounded by a small number of Dave's choice.

When the above detector is employed against the proposed DTRTF and LLCRTF schemes, there is a non-zero probability of detection owing to the fact that $M$-PSK symbols transmitted by Charlie are scaled by either $\sqrt{\alpha}$ or $\sqrt{2-\alpha}$, for some $0< \alpha < 1$. Therefore, to characterize the behaviour of their average probability of detection as a function of $\alpha$, we present the following proposition.

\begin{Proposition} \label{PD}
For the instantaneous energy detector on the $f_{CB}$ band, 
(1) the average probability of detection for the DTRTF scheme is same as that of the RHS in \cite[Theorem 2]{ISIT}.
(2) the average probability of detection for the LLCRTF scheme can be derived from that of the DTRTF, as the former scheme is a special case of the latter.
\end{Proposition}

\begin{proof}
    First, we will provide a sketch of the proof for the first part of the proposition. Note that the dummy $M$-PSK symbol in the RHS is replaced by a legitimate $M$-PSK symbol in the DTRTF scheme. Furthermore, owing to the $n$-symbol delay introduced by Charlie in the DTRTF scheme, the instantaneous energy distribution on two consecutive time-slots of the RHS over the $f_{CB}$ band is identical to that on the $k^{th}$ and $(n+k)^{th}$ time-slot over the $f_{CB}$ band in the DTRTF scheme, for $1 < k < n$. Although there is a change in the ordering of instantaneous energy levels between the two schemes, their covertness analysis are identical for a given block of $2n$ symbols. Thus, the average probability of detection for the DTRTF scheme is same as that of the RHS in \cite[Theorem 2]{ISIT}.

Within a block of $2n$ time-slots, the instantaneous energy distribution on the first $n$ time-slots of LLCRTF scheme is identical to that of DTRTF scheme. However for the $(n+k)^{th}$ time-slot, in LLCRTF scheme, Charlie chooses to scale his $M$-PSK symbol either with $\sqrt{2-\alpha}$ or unity with equal probability in an independent manner. This is in contrast to  DTRTF scheme, wherein the scale factor is chosen depending on the Alice's decoded bit that was transmitted in $k$-th time-slot. Thus, the instantaneous energy distribution on the second $n$ time-slots of the LLCRTF scheme is a deterministic variant of that on the second $n$ time-slots of the DTRTF scheme. As a result, the average probability of detection for the LLCRTF scheme can be derived from that of the DTRTF scheme. 
\end{proof}

From Proposition \ref{PD}, it is clear that the average probability of detection for the instantaneous energy detector is a function of the energy-splitting factor $\alpha$. However, we also know that the choice of $\alpha$ also dictates the average probability of decoding error at Bob. Therefore, as formulated in Problem \ref{main_opt_problem_n_time_slots}, the optimal energy-splitting factor for the DTRTF scheme is the one that minimises the sum of average probability of decoding error at Bob and average probability of detection at Dave. Given that we have upper bounds on the above terms, we propose to solve Problem \ref{main_opt_problem_n_time_slots_minima}, instead of Problem \ref{main_opt_problem_n_time_slots}, where $P_{UEavg}^{\Omega sub}$ is in \eqref{PEtotalDTRTF} and $P_{UD}^{avg}$ denotes the upper bound on the average probability of detection when using the instantaneous energy detector.

\begin{mdframed}
\begin{problem}
\label{main_opt_problem_n_time_slots_minima}
For a given noise variance N, we solve: $\alpha_{opt}^{\Omega min}= \arg \mathop {\min }\limits_{\alpha \in (0, 1)} P_{UEavg}^{\Omega sub} + P_{UD}^{avg}.$
\end{problem}
\end{mdframed}

Given the analytical complexities inherent in the derived expressions, assessing the behavior of $P_{UEavg}^{\Omega sub}$ and $P_{UD}^{avg}$ with respect to $\alpha$ is challenging and non-trivial. In light of this, we turn to simulation results, showcased in the left figure of Fig. \ref{min_int_DTRTF_LLCRTF}, which reveal that the minima of the sum of $P_{UEavg}^{\Omega sub}$ and $P_{UD}^{avg}$ closely aligns with the intersection point of $P_{avg \uparrow}^{\Omega}$ and $P_{avg \downarrow}^{\Omega}$. Here, 
\textcolor{black}{
\begin{eqnarray}{rcl}
P_{avg \uparrow}^{\Omega}\triangleq P_{avg1}^{\Omega ube} \quad \text{and} \quad
P_{avg \downarrow}^{\Omega}\triangleq P_{avg2}^{\Omega ube}+P_{UD}^{avg}
\end{eqnarray}}
are the increasing and decreasing functions of $\alpha$, respectively, where $P_{avg1}^{\Omega ube}$ and $P_{avg2}^{\Omega ube}$ are given in Theorem \ref{PEtotalDTRTFtheorem} and $P_{UD}^{avg}$ is given in \cite[Theorem 2]{ISIT}. Hence, we solve Problem \ref{main_opt_problem_n_time_slots_interesect} instead of Problem \ref{main_opt_problem_n_time_slots_minima}.
\begin{mdframed}
\begin{problem}
\label{main_opt_problem_n_time_slots_interesect}
For a given noise variance N, we solve for $\alpha_{opt}^{\Omega in}$, such that at $\alpha=\alpha_{opt}^{\Omega in}$, $P_{avg \uparrow}^{\Omega}-P_{avg \downarrow}^{\Omega}=0.$
\end{problem}
\end{mdframed}

Along the similar lines of the DTRTF scheme, as given in Problem \ref{main_opt_problem_n_time_slots_latency}, the optimal energy-splitting factor for the LLCRTF scheme is the one that minimises the sum of average probability of decoding error at Bob and average probability of detection at Dave. Therefore, we need to solve Problem \ref{main_opt_problem_n_tslatency_minima}.

\begin{mdframed}
\begin{problem}
\label{main_opt_problem_n_tslatency_minima}
For a given noise variance N, we solve: $\alpha_{opt}^{\mu min}= \arg \mathop {\min }\limits_{\alpha \in (0, 1)}P_{UEavg}^{\mu} +P_{UD}^{\mu avg}.$
\end{problem}
\end{mdframed} 

Similar to the DTRTF scheme, we take the help of simulation results, as illustrated in the right figure of Fig. \ref{min_int_DTRTF_LLCRTF}, note that the point of intersection between $P_{avg \uparrow}^{\mu}$ and $P_{avg \downarrow}^{\mu}$ closely approximates the solution offered by Problem \ref{main_opt_problem_n_tslatency_minima}, where 
\textcolor{black}{
\begin{eqnarray}{rcl}
P_{avg \uparrow}^{\mu}\triangleq P_{A5}+P_{A6},
\end{eqnarray}
\begin{eqnarray}{rcl}
P_{avg \downarrow}^{\mu} \triangleq P_{A1}^{avg}+P_{A2}^{avg}+P_{A3}^{avg}+P_{A4}^{avg}+P_{UD}^{\mu avg}.
\end{eqnarray}}
Here, $P_{A1}^{avg}$, $P_{A2}^{avg}$, $P_{A3}^{avg}$ and $P_{A4}^{avg}$ are given in Theorem \ref{PEtotalDTRTFtheorem}, $P_{A5}$ and $P_{A6}$ are given in Theorem \ref{PEtotalLLCRTFtheorem}, and $P_{UD}^{\mu avg}$ can be obtained using Proposition \ref{PD}. Consequently, we address Problem \ref{main_opt_problem_n_latencyts_interesect} as opposed to Problem \ref{main_opt_problem_n_tslatency_minima}.
 \begin{mdframed}
\begin{problem}
\label{main_opt_problem_n_latencyts_interesect}
For a given noise variance N, we solve for $\alpha_{opt}^{\mu in}$, such that at $\alpha=\alpha_{opt}^{\mu in}$, $P_{avg \uparrow}^{\mu}-P_{avg \downarrow}^{\mu}=0.$
\end{problem}
\end{mdframed}

\begin{figure}
\vspace{-0.3cm}
\includegraphics[scale=0.23]{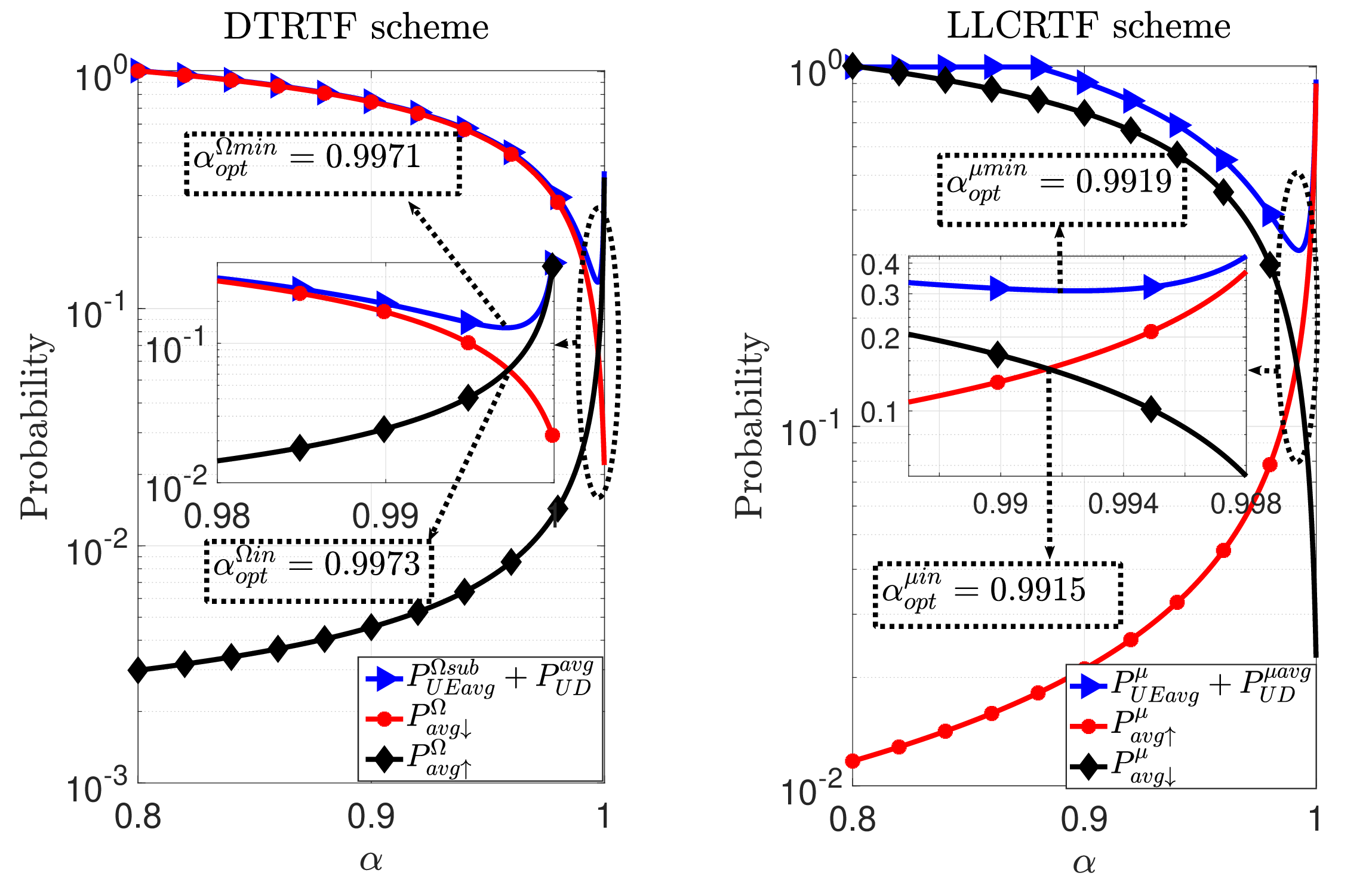}\
    \caption{\textcolor{black}{The left figure illustrates that the intersection between $P_{avg \uparrow}^{\Omega}$ and $P_{avg \downarrow}^{\Omega}$ closely approximates the minimum value of $P_{UEavg}^{\Omega sub} + P_{UD}^{avg}$ for the DTRTF scheme. We assume $N_C=1$, $\sigma^2_{AC}=4$, $\rho_{th}=10^{-5}$, and SNR $=35$ dB. The right figure shows that the intersection between $P_{avg \uparrow}^{\mu}$ and $P_{avg \downarrow}^{\mu}$ is close to the minima value of $P_{UEavg}^{\mu}+P_{UD}^{\mu avg}$ of the LLCRTF scheme for SNR $=35$ dB. }}
\label{min_int_DTRTF_LLCRTF}
\end{figure}

\begin{center}
\begin{table}[h]\caption{Solutions to Problem \ref{main_opt_problem_n_time_slots_minima}, Problem \ref{main_opt_problem_n_time_slots_interesect}, Problem \ref{main_opt_problem_n_tslatency_minima}
and Problem \ref{main_opt_problem_n_latencyts_interesect} for $N_C=1$.}\label{diff_alpha}
{%
\begin{center}
\begin{tabular}{|l|l|l|l|l|l|l|}
\hline
\multicolumn{2}{|c|}{}   & \multicolumn{2}{c|}{DTRTF} & \multicolumn{2}{c|}{LLCRTF}  \\ 
\hline
$$\text{SNR}(\text{dB})$$ & $P_{UF}^{avg}$  & $\alpha^{\Omega min}_{opt}$ & $\alpha^{\Omega in}_{opt}$ &   $\alpha^{\mu min}_{opt}$& $\alpha^{\mu in}_{opt}$  \\ \hline
\textcolor{black}{$20$ }          & \textcolor{black}{$10^{-1}$}                  & \textcolor{black}{$0.9819$}      & \textcolor{black}{$0.9903$}           & \textcolor{black}{$0.9467$}     & \textcolor{black}{$0.9560     $}  \\ \hline
$25$           & $10^{-1}$                  & $0.9933$      & $0.9970$           & $0.9796$     & $0.9842     $  \\ \hline
$30$           & $10^{-1}$                  & $0.9977$      & $0.9991$           & $0.9929$     & $0.9947     $  \\ \hline
$30$           & $10^{-2}$                  & $0.9919$      & $0.9920$           & $0.9768$     & $0.9757     $  \\ \hline
$35$          & $10^{-2}$                   & $0.9971$      & $0.9973$           & $0.9919 $ & $0.9915     $ 
\\ \hline
\end{tabular}
\end{center}
}
\end{table}
\end{center}

To validate the effectiveness of our method of choosing $\alpha$, we use numerical methods to solve Problem \ref{main_opt_problem_n_time_slots_minima}, Problem \ref{main_opt_problem_n_time_slots_interesect}, Problem \ref{main_opt_problem_n_tslatency_minima}, and Problem \ref{main_opt_problem_n_latencyts_interesect}, and list their solutions in Table \ref{diff_alpha} as a function of SNR and the expected false-positive rate of the instantaneous energy detector, denoted by $P_{UF}^{avg}$. To generate these results, we use $N_C=1$, $\sigma_{AC}^2=4$, and the step size of $\alpha$ is set to $0.0001$. The values in Table \ref{diff_alpha} indeed confirm that Problem \ref{main_opt_problem_n_time_slots_interesect} and Problem \ref{main_opt_problem_n_latencyts_interesect} can be solved in practice to recommend an appropriate value of $\alpha$ for the proposed schemes. 
\textcolor{black}{Lastly, recall that our objective is to minimize the sum of average probability of decoding  error at Bob and average probability of detection at Dave for both the schemes. Thus, it is important to discuss the impact on the choice of $\alpha$ on the sum of these probabilities. 
\textcolor{black}{In this direction, we solve Problem \ref{main_opt_problem_n_time_slots_minima} for SNR ranging from 15 dB to 40 dB, in the step size of 5 dB, for $N_C=3$, $\delta=0.842$, $\rho_{th} =10^{-5}$.
Subsequently, in Fig. \ref{SNR_15_40}, we use monte-carlo simulations to plot the sum of the above mentioned probabilities for the optimal decoder,  given in \eqref{JMAP}, and the sub-optimal decoder, given in \eqref{PesubD}, for the DTRTF scheme.} 
Furthermore, we also solve Problem \ref{main_opt_problem_n_tslatency_minima} for various SNR values and $\delta=0.842$. Next, we use these solutions to plot the sum of the probabilities for the LLCRTF scheme.
From these plots, we observe that the sum of these probabilities of both the schemes decreases with SNR, which shows the efficacy of our proposed methods. Also, we observe that the performance of the DTRTF scheme is superior to the LLCRTF scheme, owing to the fact that in the former scheme, Charlie incorporates Alice's bits into his symbols in the form of phase and energy modification, which provides an added layer of reliability to Alice's messages. However, this incorporation of Alice's messages into Charlie's symbols is impractical in the LLCRTF scheme owing to its low-latency feature.} \textcolor{black}{Finally, to showcase the effectiveness of the DTRTF strategy against various values of probability of false-alarm of the instantaneous energy detector, in Fig. \ref{ROC}, we plot the Receiver Operating Characteristic (ROC) curves \cite{ROC} for different SNRs. To generate these plots, for a given average probability of false-alarm, we solve Problem \ref{main_opt_problem_n_time_slots} for $\text{SNR}=25,30,35$ dB, and the solution to Problem \ref{main_opt_problem_n_time_slots} is given by $\alpha_{opt}^{\Omega}$.
Subsequently, we obtain the corresponding value of the average probability of detection for $\alpha_{opt}^{\Omega}$.
From these plots, we observe that the ROC curves of the countermeasure do not appear bulged outwards away from the random-classifier, thereby not providing high probability of detection for low probability of false alarm. This confirms the effectiveness of the proposed countermeasures against the instantaneous energy detector.
}

\begin{figure}[ht!]
\vspace{-0.1cm}
\begin{center}
    
\includegraphics[scale=0.20]{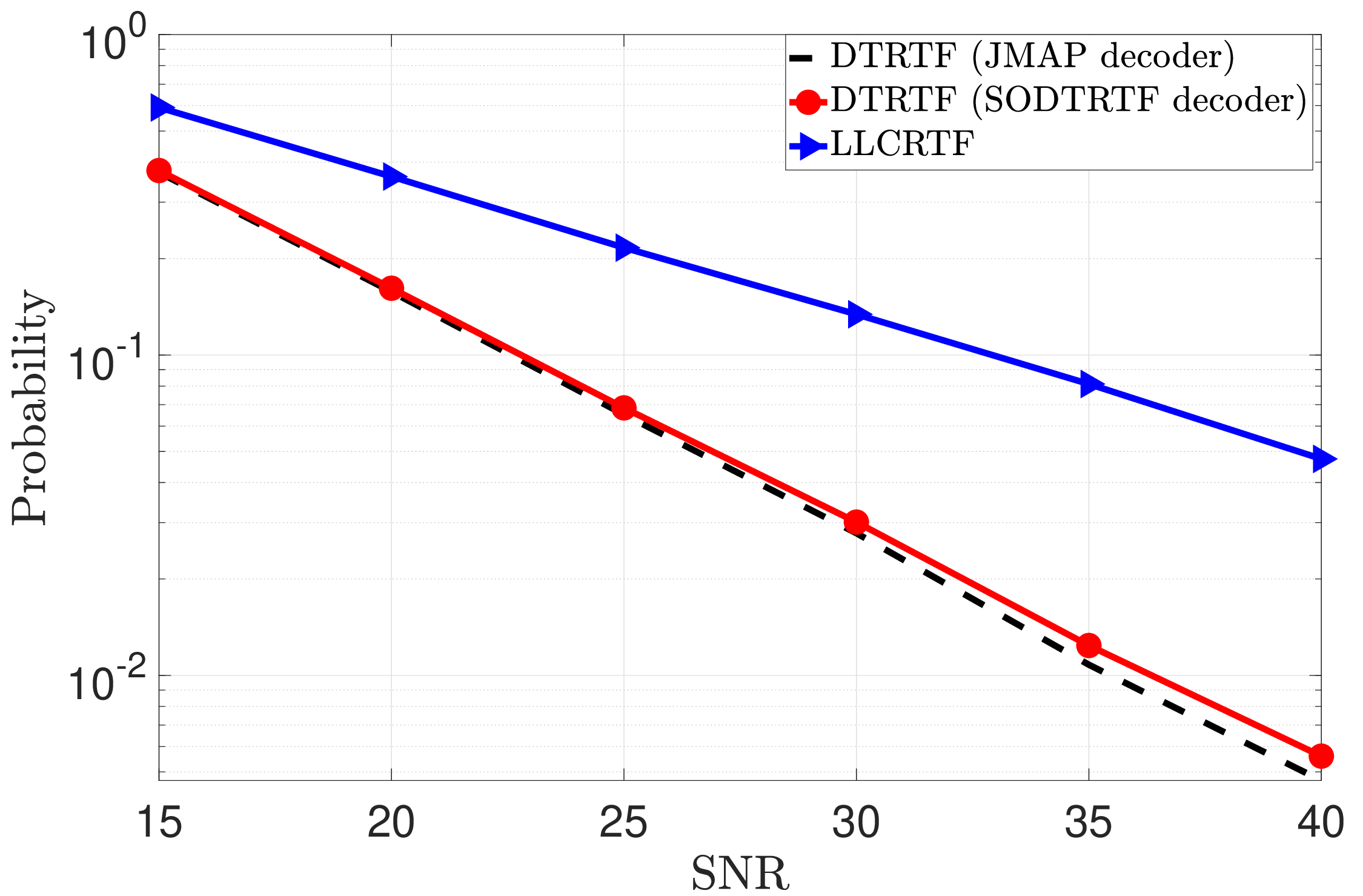}
    \caption{\textcolor{black}{Plots on the sum of probability of decoding error at the destination and  probability of detection at the adversary, as a function of SNR, for the DTRTF and LLCRTF scheme.} }
    \label{SNR_15_40}
\end{center}

\vspace{-0.3cm}
\end{figure}

\begin{figure}[ht!]
   \begin{center}
       {\includegraphics[scale=0.20]{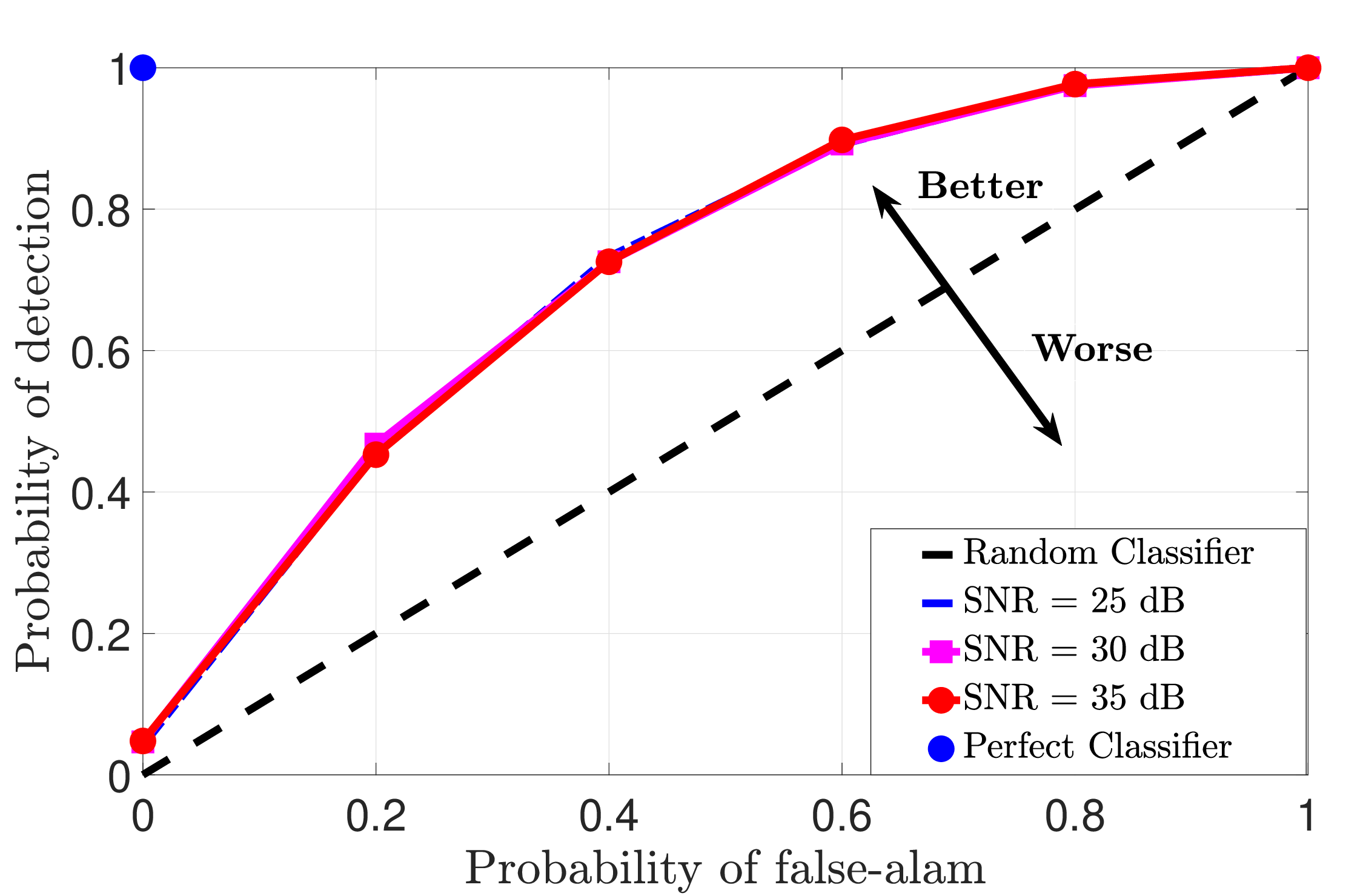}}
\caption{\textcolor{black}{ROC curves for the instantaneous energy detector at Dave, for different SNRs.}} 
\label{ROC}
  \end{center}  
    \end{figure}

Next, we will analyse the \textcolor{black}{covertness} against the KLD-estimator based detector on 
the $f_{CB}$ bands.

\subsection{KLD-Estimator Based Detector}
\label{sec:kld_fcb}

First, we will discuss the covertness of our schemes against the KLD-estimator based detector on the $f_{CB}$ band. When using the KLD-estimator based detector on the $f_{CB}$ band, it is clear that the statistical distribution on the energies of the received symbols after the countermeasure is not identical to that before the countermeasure. As a result, it is important to study the performance of the KLD-estimator based detector when the DTRTF and LLCRTF schemes are deployed for various values of $\alpha \in (0, 1)$.  However, given the short frame lengths of the packets, it is well known that analytically characterising the performance of KLD-estimator is challenging \cite{V3,V4}. As a consequence, we use simulation results to present the performance of the KLD-estimator based detector by assuming that the choice of $\alpha$ for the DTRTF and LLCRTF schemes is already optimized against the instantaneous energy detector. In this context, we solve Problem \ref{main_opt_problem_n_time_slots_interesect} for DTRTF scheme and Problem \ref{main_opt_problem_n_latencyts_interesect} for LLCRTF scheme with the parameters $N_C=1$, SNR of $35$ in dB. Subsequently, we compare the statistical distribution of the received symbols at Dave before and after the countermeasure for these solutions using KLD-estimator based detector \cite{KLD}, and present their KLD estimates in Fig. \ref{KLDgraphfcb} along with the plots corresponding to $\alpha=1$. In the same figure, we also plot the KLD estimates of the RHS when using its appropriate value of $\alpha$ denoted by $\alpha_{opt}^{in}$. 
From Fig. \ref{KLDgraphfcb}, we infer that the KLD estimates for all the schemes are close to zero, which signifies that the statistical distribution of the received symbols on $f_{CB}$ are almost identical before and after the countermeasure for the solutions provided by Problem \ref{main_opt_problem_n_time_slots_interesect} and Problem \ref{main_opt_problem_n_latencyts_interesect}. In practice, Dave will need to set a threshold on the KLD estimate depending on a tolerable false-positive rate, and then raise a detection event when the estimate under inspection is above the threshold. From our results in Fig. \ref{KLDgraphfcb}, it is clear that for the detection of our schemes, even if Dave sets a threshold, the probability of detection will be very small as the KLD estimates for all schemes are close to the one with $\alpha = 1$.

\begin{figure}[ht!]
\includegraphics[scale=0.23]{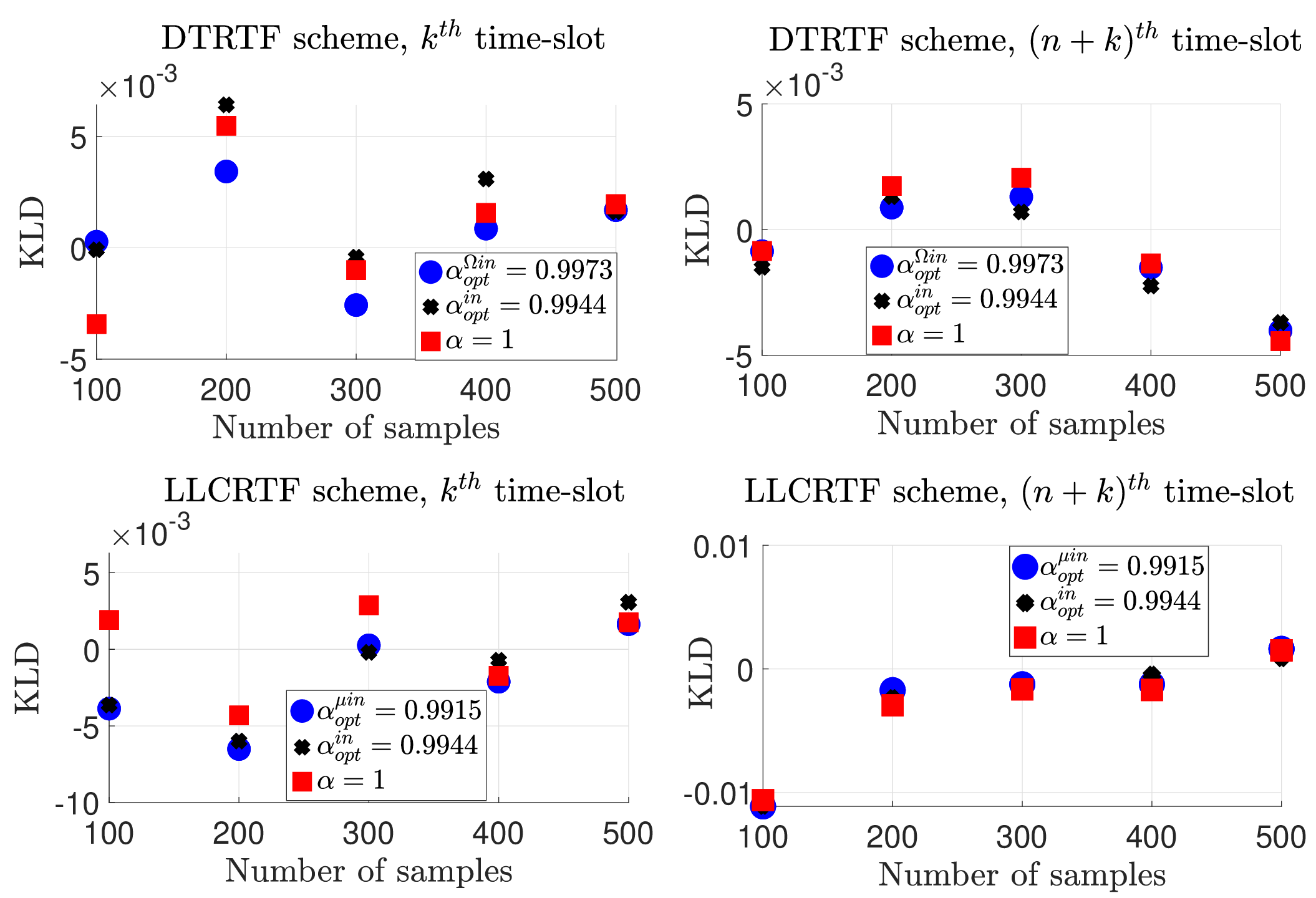}
\caption{\textcolor{black}{Average KLD estimates on $f_{CB}$ for LLCRTF, DTRTF and RHS, for $N_{C}=1$, and SNR  $=35$ dB.}}
\label{KLDgraphfcb}
\end{figure}

\begin{figure}[ht!]%
\vspace{-0.3cm}
\begin{center}
\includegraphics[scale=0.20]{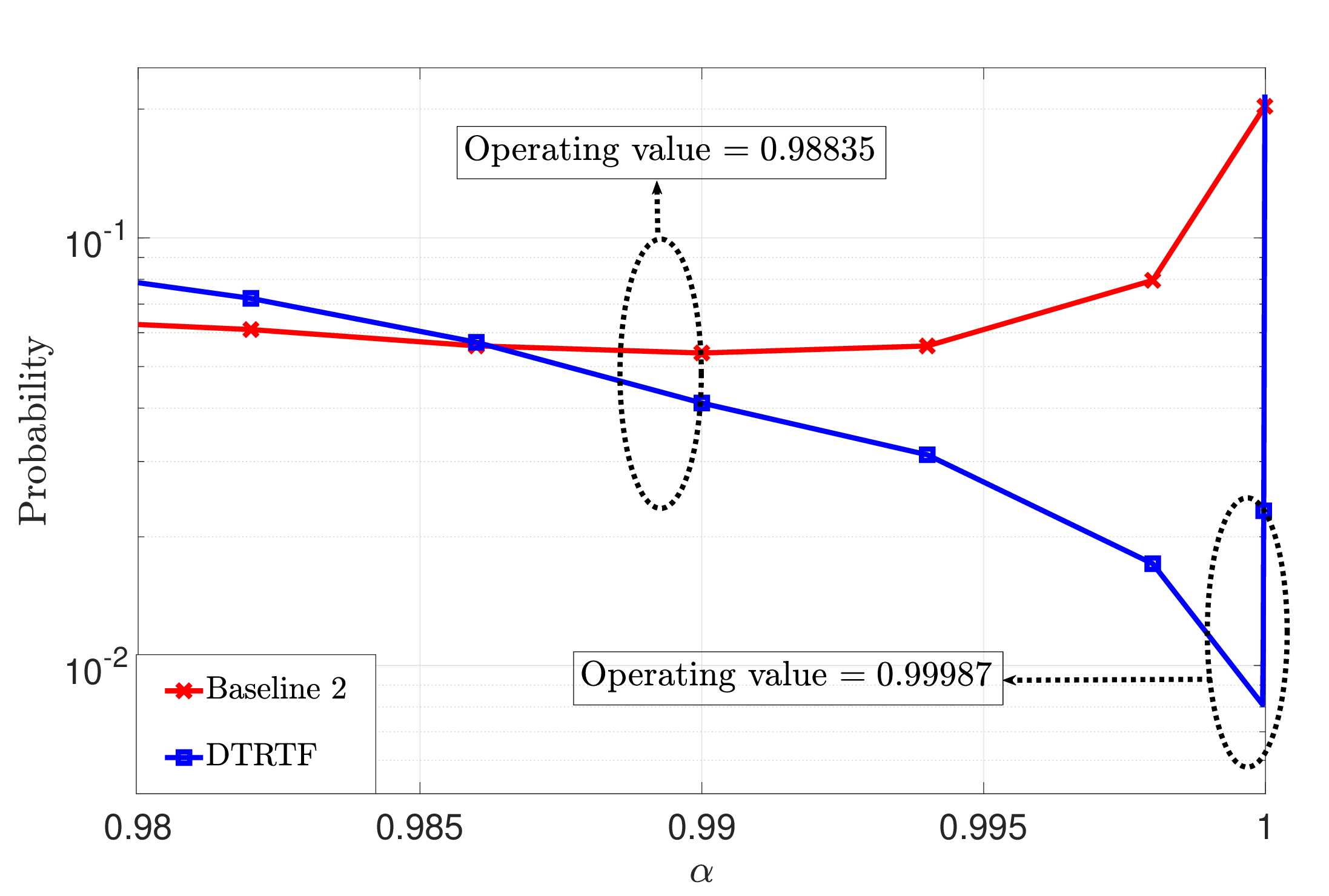}
    \caption{\textcolor{black}{Sum of the probability of decoding error and the probability of detection for the DTRTF scheme and the countermeasure in \cite{lowlatency}.}}
    \label{PePdref3DTRTF}%
        
\end{center}
\vspace{-0.5cm}
\end{figure}

\section{Comparison with Baselines} \label{SR}
\textcolor{black}{ 
We highlight that \cite{V2} and \cite{lowlatency}  have considered threat models wherein the  reactive adversary only measures the average energy on its observations on the victim's frequency band. In contrast, we and the authors in \cite{ISIT,TVT1} consider a stronger threat model, wherein the adversary measures various generalized energy statistics such as the instantaneous energies and the statistical distribution of the energies on all the frequency bands.
As a result, the countermeasures in \cite{V2} and \cite{lowlatency} are not tailor-made for our threat model. Additionally, the countermeasures proposed in \cite{V2} assumes instantaneous decoding of the victim's bit, which is difficult to achieve with practical FDRs, therefore, we will first compare the countermeasure proposed in \cite{lowlatency} with our proposed strategies for the threat model proposed in this work. 
In this direction, in Fig. \ref{PePdref3DTRTF}, we use monte-carlo simulations to plot the sum of probability of decoding error at Bob and the probability of detection at Dave as a function of energy-splitting factor $\alpha$, for $Eb/No$ of $34$ dB (owing to difference in spectral efficiencies), $N_C=190$, $\delta=0.49$, for the DTRTF scheme. Recall that the operating value of $\alpha$ is the value of $\alpha$ for which the sum of the probabilities is minimum. 
From the plots, we observe at the operating value of the countermeasure in \cite{lowlatency}, the value of the sum of the probabilities is strictly higher than that for the DTRTF scheme, thus justifying the need for designing our countermeasure despite   \cite{lowlatency}.}   

\color{black}

As the threat model in our work is same as that in \cite{ISIT,TVT1}, we will now compare our proposed strategies with that of \cite{ISIT,TVT1}. In the preceding section, we have presented methods to choose the energy splitting factor $\alpha$ based on the joint error probability of both the victim and the helper node. Now, we will discuss the impact of such a choice of $\alpha$ on their individual error performance. To generate the simulation results, we use the system model and channel parameters specified in Section \ref{SM}, Section \ref{DTMS} and Section \ref{LLCA}. 
\textcolor{black}{To present the results for DTRTF scheme, we solve Problem \ref{main_opt_problem_n_time_slots_minima} for SNR ranging from $28$ dB to $40$ dB, in the step size of $1$ dB, for $N_C=3$, $\delta=0.483$ and $\rho_{th}=10^{-5}$.}
Then, we use the solutions of Problem \ref{main_opt_problem_n_time_slots_minima} and run monte-carlo simulations to obtain the individual average error performance using the optimal decoder, given in \eqref{JMAP}, and the sub-optimal decoder, given in \eqref{PesubD}. Likewise, for  LLCRTF scheme, we solve Problem \ref{main_opt_problem_n_tslatency_minima} for SNR ranging from $28$ dB to $40$ dB, in the step size of $1$ dB and $\delta=0.483$. Subsequently, we use these solutions to obtain the individual error performance of the two users. Additionally, we also compare the individual error performance of the DTRTF and LLCRTF schemes with that in RHS \cite{ISIT,TVT1}. If we directly use the framework of RHS, the comparison is not fair as the rate of RHS is half whereas the rates of DTRTF and LLCRTF are three-fourth. As a result, for a fair comparison, we use a variant of  RHS, wherein we replace the dummy symbol in RHS with a legitimate $M$-PSK symbol, thereby keeping identical rates. In this variant of RHS, in the first time-slot, Bob treats Alice's transmitted symbol as interference and then decodes Charlie's transmitted symbol. Then, he removes the decoded $M$-PSK symbol from the received symbol before implementing the JMAP and Joint Dominant (JD) decoders in \cite{ISIT,TVT1}. Using the above description, we present the individual error performance of the two users for all the three schemes in  Fig. \ref{Performance comparison}, wherein, the left figure and the right figure capture the average symbol error probabilities of Alice and Charlie, respectively.

From the left figure of Fig. \ref{Performance comparison}, we infer that Alice's error performance when using the DTRTF scheme is poorer than that using the RHS. This is because, for a given SNR, in DTRTF scheme, the behaviour of joint probability of error is such that the solution of Problem \ref{main_opt_problem_n_time_slots_minima} is closer to unity, as compared to that of RHS. 
As a result, for the DTRTF scheme, the energy difference between bit-1 and bit-0 in the $k^{th}$ time-slot decreases, and so is the distance between the concentric circles formed by the two sets of $M$-PSK symbols during the $(n+k)^{th}$ time-slot. Therefore, the individual error performance for Alice's bits in the DTRTF scheme is poorer as compared to that in the RHS. From the right figure of Fig. \ref{Performance comparison}, we infer that Charlie's error performance in the DTRTF scheme is better than that of the RHS. This is because, along the similar lines of the left figure of Fig. \ref{Performance comparison}, in the DTRTF scheme, the solution to Problem \ref{main_opt_problem_n_time_slots_minima} is closer to unity, compared to that in the RHS, As a consequence, the envelope of the $M$-PSK symbols transmitted during the first time-slot shrinks. Secondly, as Alice's bits are considered as interference in the RHS, the overall interference on the first time-slot at Bob increases. As a consequence of the above points, the overall SINR for the $M$-PSK symbol in the first time-slot, decreases. Consequently, the individual error performance of Charlie for the RHS is poorer than that of DTRTF scheme. \textcolor{black}{Thus, the DTRTF scheme is a helper-friendly scheme, and is applicable when the helper node neither wants to compromise on their rate nor their error performance.}

\begin{figure}%
\begin{center}
\includegraphics[scale=0.23]{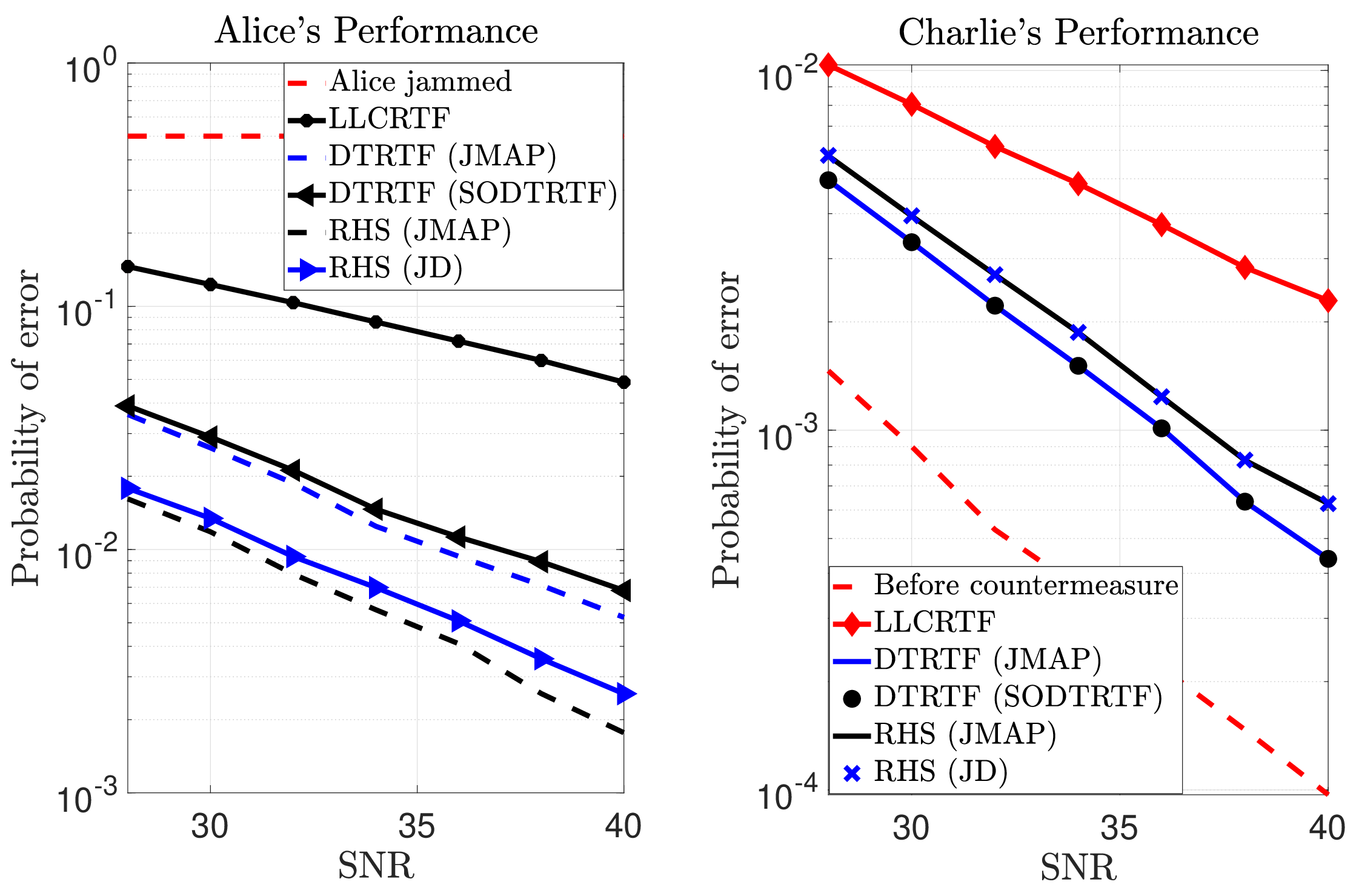}
  
    \caption{\textcolor{black}{Average probability of decoding error associated with DTRTF and LLCRTF and RHS, for
    Alice's bits (shown by the left figure), and Charlie's symbols (shown by the right figure).}}
\label{Performance comparison} 
        
\end{center}
\vspace{-0.5cm}
\end{figure}

\color{black}
\section{Robustness of the Proposed Strategies} \label{Robustness}
In this section, we will discuss the 
 robustness of the proposed strategies by considering QAM at the helper node, and arbitrary position of the helper node with respect to Alice and Bob.\footnote{\textcolor{black}{Alice continues to transmit her information modulated using OOK, and her relative position with respect to Bob is fixed.}} In this direction, we assume that Frank acts as the helper and  uses $M$-QAM to communicate his information symbols to Bob. As posed in Problem \ref{main_opt_problem_n_time_slots}, the performance of the strategies with QAM at the helper node is jointly determined by the probability of decoding error at the destination and the probability of detection at the adversary. Therefore, we fix SNR and vary $\alpha$, and obtain the minimum value of the sum of these probabilities using monte-carlo simulations.\footnote{\textcolor{black}{As the DTRTF scheme outperforms the LLCRTF scheme, we compare the performance of the former scheme with PSK and QAM.}}

To present the performance of the DTRTF scheme with $M$-QAM, in the leftmost figure and the middle figure of Fig. \ref{diff_variance}, we plot the sum of the probability of decoding error at Bob and probability of detection at Dave as a function of SNR for $M$-QAM, corresponding to the plots of $M$-PSK. 
From these plots, we observe that for a given SNR, the proposed DTRTF scheme provides better performance with $M$-QAM than with $M$-PSK.
This is because, for a given SNR,  the probability of error and the probability of detection for $M$-QAM are less than that those for $M$-PSK, consequently, the sum of the probabilities is also less for $M$-QAM compared to that for $M$-PSK.

\begin{figure}[ht!]
\begin{center}
\includegraphics[scale=0.23]{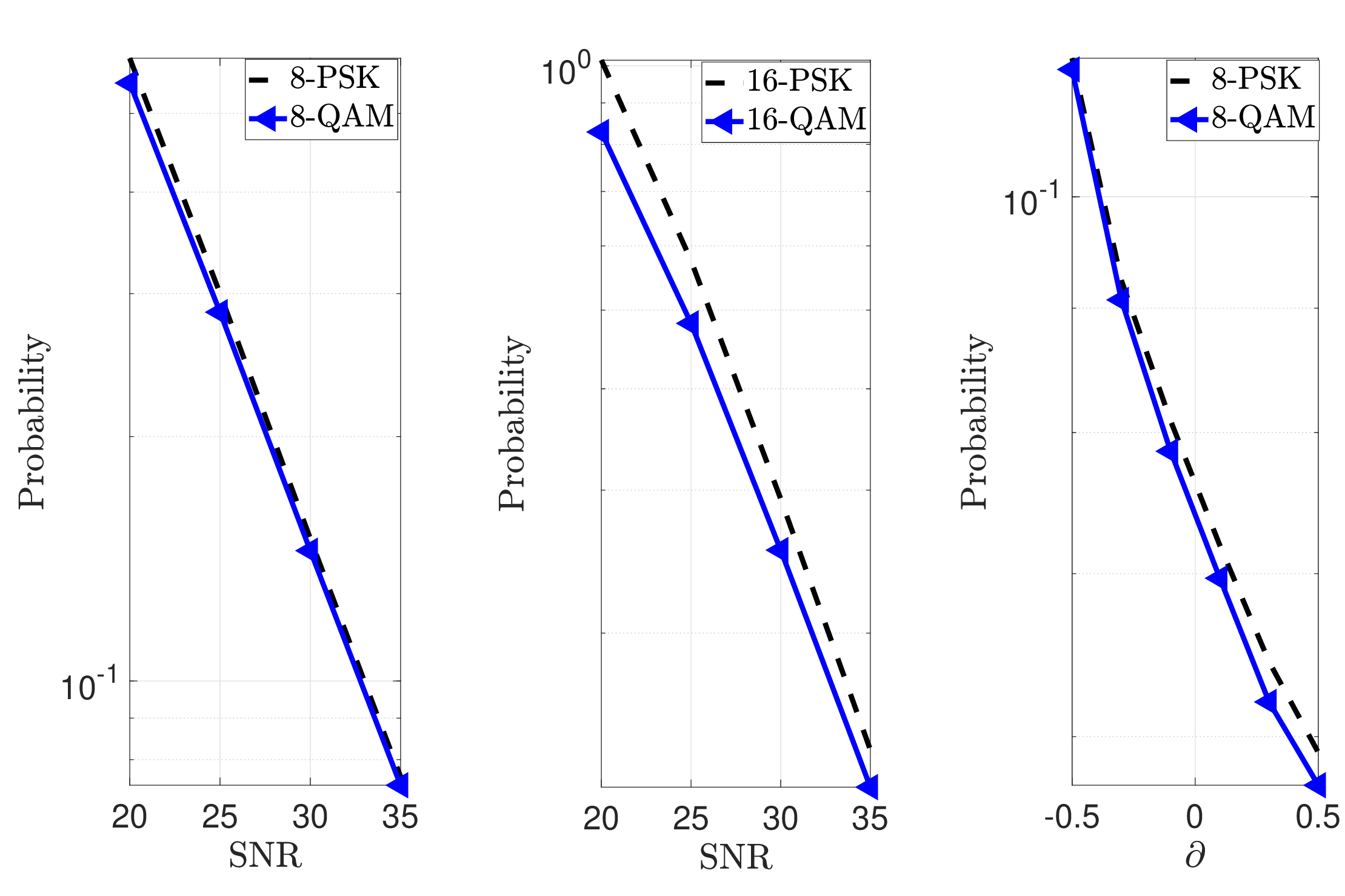}
   \vspace{-0.1cm}
    \caption{\textcolor{black}{DTRTF scheme performance against QAM at the helper node and arbitrary position of the helper node.}}
  \label{diff_variance}%
        
\end{center}
\end{figure}

When the channel variances of the link between Alice-Bob and Helper-Bob are identical, i.e., $\partial=0$, the proposed countermeasures were proven to maintain all the generalized energy statistics on both the frequency bands.
However, in practice, large scale fading on these two channels may not be always identical, i.e., $\partial \neq 0$.
In this direction, in the rightmost figure of Fig. \ref{diff_variance}, we use monte-carlo simulations to plot the sum of probabilities as a function of $\partial$ for $M=8$, at SNR $=35$ dB. 
From this figure, we observe that as $\partial$ increases, i.e., the helpers move closer to the centre, the sum of the probabilities decreases. This is because, as $\partial$ increases, the reliability of the link between Charlie-Bob and Frank-Bob improves, as a result, the probability of error decreases. Additionally, following the footsteps of Fig. \ref{KLDgraphfcb}, in Fig. \ref{KLDgraphfab2}, we plot the KLD estimates for $\partial = 0.5$ for the $f_{AB}$ band, and  we observe that the KLD estimates are close to zero, which signifies that the statistical distribution of the energies of the received symbols are almost identical before and after the countermeasure, even though $\alpha$ is optimised for $\partial=0$. 

\begin{figure}[ht!]
\begin{center}
\includegraphics[scale=0.23]{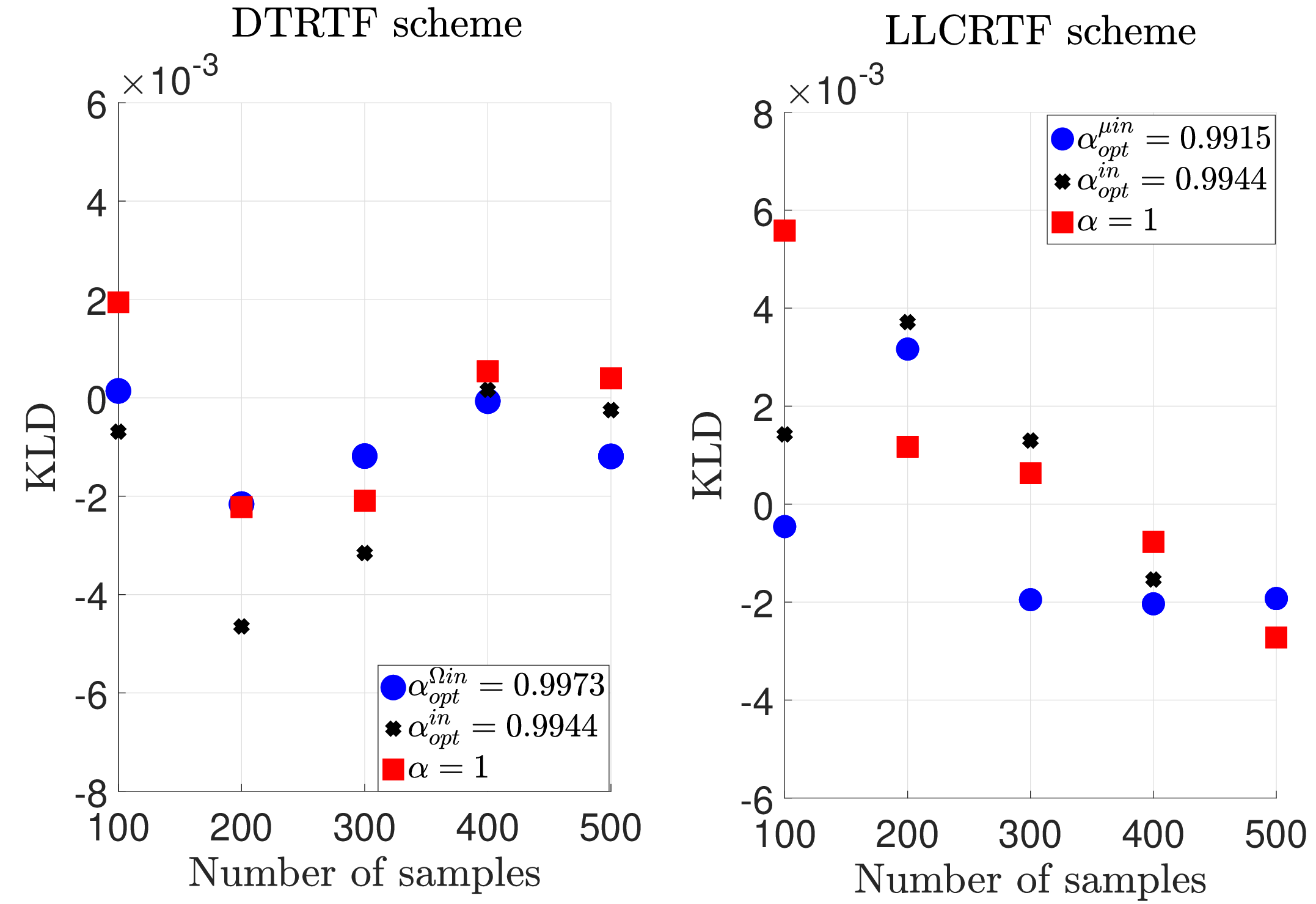}
   
    \caption{\textcolor{black}{Average KLD estimates on $f_{AB}$ for the DTRTF and the LLCRTF for $\partial=0.5$.}}
     \label{KLDgraphfab2}%
        
\end{center}
\vspace{-0.5cm}
\end{figure}

\section{Future Research Directions}\label{SS}

\textcolor{black}{
From the standpoint of wireless channel conditions, our countermeasures are applicable when the helper's channel is quasi-static with arbitrary coherence-time as long as coherent detection can be implemented at Bob. However, for fast-fading channels, wherein the coherence-time is one time-slot, it is well known that non-coherent-modulation schemes are applicable. Consequently, there is a need to design countermeasures using non-coherent modulation for such scenarios. From the standpoint of the adversary, in this work, Dave uses a threshold based technique on the energies of the observed samples to monitor the instantaneous energy of the transmitted symbols on all the bands. Another interesting direction for future research will be to explore machine-learning based methods \cite{ML1} for designing improved instantaneous energy detectors. Subsequently, by using the results on the probability of detection for such detectors, one can arrive at near-optimal energy-splitting factors by using our results on the error performance of the countermeasures.}

\color{black}

\end{document}